\documentclass[11pt]{article}

\usepackage{amsmath}
\usepackage{dsfont}
\usepackage{ntheorem}
\usepackage{graphicx}
\usepackage{amssymb}
\usepackage{color}
\usepackage{authblk}
\usepackage{fullpage}

\newenvironment{proof}{
\noindent {\em Proof.}\ }{{\hfill\
    $\Box$}\vspace{.3pc}}

\newcommand{\diam}{\mathcal D}
\newcommand{\algoU}{{\tt Algo}_{\mathcal U}}
\newcommand{\algoHC}{{\tt Algo}_{\mathcal {HC}}}
\newcommand{\algoFHC}{{\tt Algo}_{\mathcal {FHC}}}
\newcommand{\algoB}{{\tt Algo}_{\mathcal B}}

\newcommand{\neig}{\mathcal N}
\newcommand{\mind}{Min\_d}

\newcommand{\Best}{bestParent}
\newcommand{\update}{update}
\newcommand{\correctb}{Pred\_correct\_node}
\newcommand{\correctc}{Pred\_correct\_d}
\newcommand{\correcta}{Pred\_SUB\_d}
\newcommand{\correctx}{Pred\_UB\_d}
\newcommand{\Root}{\textit{R}} 
\newcommand{\UNBOUNDED}{VUB}

\newtheorem{theorem}{Theorem}
\newtheorem{lemma}{Lemma}
\newtheorem{remark}{Remark}
\newtheorem{corollary}{Corollary}
\newtheorem{definition}{Definition}
\newtheorem{observation}{Observation}
\newtheorem{notation}{Notation}

\newcommand{\cache}[1]{}

\newcommand{\Ca}{conf1} 
\newcommand{\Cb}{conf2} 
\newcommand{\Cc}{conf3} 
\newcommand{\Cd}{conf4}
\newcommand{\const}{nbSteps} 
\newcommand{\nbTotal}{nbTotal}

\newcommand{\D}{D} 
\newcommand{\B}{Att\_HC}
\newcommand{\A}{Att\_B}
\newcommand{\CD}{Att\_dist}
\newcommand{\UBD}{Att\_UB}
\newcommand{\UB}{Att\_SUB}
\newcommand{\LA}{Att}

\begin{document}
\title{Silent Self-stabilizing BFS Tree Algorithms Revised}
%\subtitle{Sous-titre \\ ici}
\author{St\'ephane Devismes}
\affil{VERIMAG}
\author{Colette Johnen}
\affil{Univ. Bordeaux, LaBRI, UMR 5800, F-33400 Talence, France}

\date{}
\maketitle

\begin{abstract}
  In this paper, we revisit two fundamental results of the
  self-stabilizing literature about silent BFS spanning tree
  constructions: the Dolev {\em et al} algorithm and the Huang and
  Chen's algorithm. More precisely, we propose in the composite
  atomicity model three straightforward adaptations inspired from
  those algorithms. We then present a deep study of these three
  algorithms.  Our results are related to both correctness
  (convergence and closure, assuming a distributed unfair daemon) and
  complexity (analysis of the stabilization time in terms of rounds
  and steps).
\end{abstract}

\paragraph{Keywords:} Self-stabilization, BFS spanning tree,
composite atomicity model, distributed unfair daemon, stabilization
time, round and step complexity.

%\tableofcontents

\section{Introduction}

\textit{Self-stabilization}~\cite{Dijkstra74} is a versatile technique
to withstand \textit{any finite number} of transient faults in a
distributed system: a self-stabilizing algorithm is able to recover a
correct behavior in finite time, regardless of the \emph{arbitrary}
initial configuration of the system, and therefore, also after the
occurrence of transient faults.

After the seminal work of Dijkstra, several self-stabilizing
algorithms have been proposed to solve various tasks such as token
circulations~\cite{HuangC93}, clock
synchronization~\cite{CouvreurFG92}, propagation of information with
feedbacks~\cite{BuiDPV99}, {\em etc}. Among the vast self-stabilizing
literature, many works more precisely focus on the construction of
distributed data structures, {\em e.g.}, minimal dominating
sets~\cite{KakugawaM06}, clustering~\cite{CaronDDL10}, spanning
trees~\cite{CYH91}.
Most of the self-stabilizing algorithms which construct distributed
data structures actually achieve an additional property called {\em
  silence}~\cite{DolevGS96}: a silent self-stabilizing algorithm
converges within finite time to a configuration from which the value
of all its communication variables are constant.

\paragraph{Related Works.} We focus here on silent self-stabilizing
spanning tree constructions, {\em
  e.g.},~\cite{CYH91,HC92,DIM93,CollinD94,KK05,CRV11}. Spanning tree
constructions are of major interest in networking, {\em e.g.}, they
are often involved in the design of routing and broadcasting tasks.
Moreover, (silent) self-stabilizing spanning tree constructions are
widely used as a basic building blocks of more complex
self-stabilizing solutions. Indeed, {\em composition} is a natural way
to design self-stabilizing algorithms~\cite{T01} since it allows to
simplify both the design and proofs of self-stabilizing algorithms.
Various composition techniques have been introduced so far, {\em
  e.g.}, collateral composition~\cite{GoudaH91}, fair
composition~\cite{D00b}, and conditional
composition~\cite{DattaGPV01}; and many self-stabilizing algorithms
actually are made as a composition of a silent spanning tree algorithm
and another algorithm designed for tree topologies.
For example, collateral, fair, and conditional compositions
are respectively used the design of the algorithms given
in~\cite{Lin},~\cite{ButelleLB95}, and~\cite{CournierDV05}.  Notably,
the silence property is not mandatory in such designs, however it
allows to write simpler proofs~\cite{DattaLDHR13}.

Many self-stabilizing spanning tree constructions have been proposed,
{\em e.g.},~\cite{CYH91,HC92,DIM93,CollinD94,KK05,CDV09,CRV11}. These
constructions mainly differ by the type of tree they compute, {\em
  e.g.}, the tree can be arbitrary~\cite{CYH91},
depth-first~\cite{CollinD94},
breadth-first~\cite{HC92,DIM93,CDV09,CRV11}, {\em etc}.
In this paper, we focus two particular {\em Breadth-First Search}
(BFS) spanning tree constructions: the one of Huang and
Chen~\cite{HC92}, and the one of {\em Dolev et al}~\cite{DIM93}. These
two constructions are among the most commonly used in the
self-stabilizing literature.\footnote{As a matter of facts,
  \cite{HC92} and~\cite{DIM93} are respectively cited 109 and 409
  times in \texttt{Google Scholar}.} Indeed, these constructions
cumulate several advantages:
\begin{enumerate}
\item Their design is simple.
\item The BFS spanning tree is really popular because of its
  minimum height.
\item They are silent. Notice, by contrast, that the solution given
  in~\cite{CDV09} is not silent.
\item Despite their time complexity was not analyzed until now, they
  are commonly assumed to be asymptotically optimal in rounds, {\em
    i.e.}, $O(\diam)$ rounds, where $\diam$ is the diameter of the
  network. Notice, by contrast, that the stabilization time of the
  solutions proposed in~\cite{CDV09} and~\cite{CRV11} are
  $O(n+\diam^2)$ rounds and $O(\diam^2)$ rounds, respectively.
\end{enumerate}
More precisely, the Huang and Chen's algorithm~\cite{HC92} is written
in the composite atomicity model. It assumes the processes have the
knowledge of $n$, the size of the network. This assumption allows
processes to have a bounded memory: $\Theta(\log n)$ bits are required
per process.  The algorithm is proven assuming a {\em distributed
  unfair daemon}, the most general scheduling assumption of the
model. However, no complexity analysis is given about its
stabilization time in steps or rounds, the two main complexity
metrics of the model.

The Dolev {\em et al}'s algorithm~\cite{DIM93} is written in the
read/write atomicity model. This model is more general than the
composite atomicity model.  The algorithm does not assume any
knowledge on any global parameter of the network, such as $n$ for
example. The counterpart being that there is no bound on process local
memories.  The algorithm is proven under the central fair assumption
({\em n.b.}, the notion of unfair daemon is meaningless in this
model). Despite no complexity analysis is given in the paper, authors
conjecture the stabilization time is asymptotically optimal, {\em
  i.e.}, $O(\diam)$ rounds. By definition, a straightforward
translation of the algorithm of Dolev {\em et al} also works in the
composite atomicity model assuming a distributed weakly fair
daemon. However, an ad hoc proof is necessary if one want to establish
its self-stabilization under a distributed unfair daemon.
Notice that the algorithm of Dolev {\em et al} and its bounded memory
variant are used in the design of several algorithms written in this
composite atomicity model, {\em e.g.},~\cite{AroraG94,Karaata02}.
Moreover, several algorithms are based on similar principles, {\em
  i.e.},~\cite{DLV11j0,GHIJ14}.  Hence, a proof of its
self-stabilization assuming a distributed unfair daemon is highly
desirable.

\paragraph{Contribution.} In this paper, we study three silent
self-stabilizing BFS spanning tree algorithms written in the composite
atomicity model.  The first algorithm, called $\algoU$, is the
straightforward translation of the Dolev {\em et al}'s
algorithm~\cite{DIM93} into the composite atomicity model. The second
one, $\algoB(D)$, is a variant of $\algoU$, where the process local
memories are bounded. To that goal, the knowledge of some upper bound
$D$ on the network diameter is assumed. Finally, The third algorithm,
noted $\algoHC(D)$, is a generalization of the Huang and Chen's
algorithm~\cite{HC92}, where the exact knowledge of $n$ is replaced by
the knowledge of some upper bound $D$ on the network diameter ({\em
  n.b.}, by definition, $n$ is a particular upper bound on the
diameter).

The general purpose of this paper is twofold. First, we show the close
relationship between these three algorithms. To see this, we propose
a general and simple proof of self-stabilization for the
three algorithms under the distributed unfair daemon. This proof
implies that every executions of each of the three algorithms is
finite in terms of steps.  Moreover, notice that the proof shows in
particular that the assumption on the exact knowledge of $n$ in the
initial algorithm of Huang and Chen was too strong. Second, the proof
also validates the use of the Dolev {\em et al}'s algorithm and its
bounded-memory variant in the composite atomicity model assuming any
daemon (in particular, the unfair one).

Second, we propose a complexity analysis the stabilization time of
these three algorithms in both steps and rounds. Our results are both
positive and negative.  First, we show that the stabilization time of
$\algoU$ and $\algoB(D)$ is {\em optimal} in rounds by showing that in
both cases the worst case is exactly $\diam$ rounds. With few
modifications our proof can be adapted for the read/write atomicity
model, validating then the conjecture in~\cite{DIM93} which claimed
that the stabilization time the Dolev {\em et al}'s algorithm was
asymptotically optimal in rounds.

We then establish a lower bound in $\Omega(D)$ rounds on the
stabilization time of $\algoHC(D)$. Now, $\algoHC(n)$ is actually the
algorithm of Huang and Chen. Thus, the algorithm of Huang and Chen
stabilizes in $\Omega(n)$ rounds. This result may be surprising as
until now this algorithm was conjectured to stabilize in $O(\diam)$
rounds. More precisely, this negative result is mainly due to the fact
that two rules of the algorithm are not mutually exclusive, and when
both rules are enabled at the same process $p$, the daemon may choose
to activate any of them. Our lower bound is thus established when the
daemon gives priority to one of the two rules ($HC_1$). Hence, to
circumvent this problem we proposed a straightforward variant, noted
$\algoFHC(D)$, where we give priority to the other rule ($HC_2$). We
then establish that in this latter case the stabilization time becomes
exactly $\diam+1$ rounds in the worst case.

Finally, we consider the stabilization time in steps. Our results are
all negative. Indeed, we first show that the stabilization time in
steps of $\algoU$ cannot be bounded by any function depending on
topological parameters. We then exhibit a lower bound exponential in
$\diam$ (the actual diameter of the network) on the stabilization time
in steps which holds for both $\algoB(D)$ and $\algoHC(D)$.  Notice,
by contrast, that the stabilization time of the solutions proposed
in~\cite{CDV09} and~\cite{CRV11} are $O(\Delta.n^3)$ steps
and~$O(n^6)$ steps, respectively.

\paragraph{Roadmap.} The rest of the paper is organized as
follows. The next section is dedicated to the description of the
computation model and definitions. The formal codes of the three
algorithms are given in Section~\ref{sect:algo}. In
Section~\ref{sect:corr}, we propose our general proof of
self-stabilization assuming a distributed unfair daemon. In
Section~\ref{sect:round}, we analyze the stabilization time in rounds
of the three solutions. We present an analysis of the stabilization
time in steps of the three solutions in
Section~\ref{sect:step}. Section~\ref{sect:ccl} is dedicated to
concluding remarks.

\section{Preliminaries}

\subsection{Distributed Systems}

We assume distributed systems of $n>0$ interconnected processes. One
process, called the {\em root}, is distinguished, the others are {\em
  anonymous}. The root of the system is simply denoted by $\Root$.
Each process $p$ can directly communicate with a subset $\neig_p$ of
other processes, called its {\em neighbors}.  We assume bidirectional
communications, {\em i.e.}, if $q \in \neig_p$, then $p \in
\neig_q$. The topology of the system is a simple undirected connected
graph $G = (V,E)$, where $V$ is the set of processes and $E$ is the
set of edges, each edge being an unordered pair of neighboring
processes.   
$\|p,q\|$ denotes the distance (the length of the
shortest path) from $p$ to $q$ in $G$. We denote by $\diam$ the
diameter of $G$, {\em i.e.}, $\diam= \max_{p,q\in V} \|p,q\|$.

\subsection{Computational Model} We consider the {\em locally shared
  memory model\/ with composite atomicity} introduced by
Dijkstra~\cite{Dijkstra74}, where each process communicates with its
neighbors using a finite set of {\em locally shared variables},
henceforth called simply {\em variables}.  Each process can read its
own variables and those of its neighbors, but can write only to its
own variables.  Each process operates according to its (local) {\em
  program}.  A \emph{distributed algorithm\/} ${\mathcal A}$ consists
of one local program $\mathcal A(p)$ per process $p$.

$\mathcal A(p)$ is given as a finite set of \textit{rules}: $\{Label_i
: Guard_i \to Action_i\}$.  {\em Labels} are only used to identify
rules in the reasoning.  The {\em guard\/} of a rule in $\mathcal
A(p)$ is a Boolean expression involving the variables of $p$ and its
neighbors. The {\em action\/} part of a rule in $\mathcal A(p)$
updates some variables of $p$.  The \emph{state} of a process in
${\mathcal A}$ is defined by the values of its variables in ${\mathcal
  A}$.  A {\em configuration\/} of ${\mathcal A}$ is an instance of
the states of every process in ${\mathcal A}$.  $\mathcal C_{\mathcal A}$ is
the set of all possible configurations of ${\mathcal A}$. (When there
is no ambiguity, we omit the subscript ${\mathcal A}$.)  A rule can be
executed only if its guard evaluates to {\em true}; in this case, the
rule is said to be {\em enabled\/}.  A process is said to be enabled
if at least one of its rules is enabled.  We denote by $\mbox{\it
  Enabled}(\gamma)$ the subset of processes that are enabled in
configuration $\gamma$.  When the configuration is $\gamma$ and
$\mbox{\it Enabled}(\gamma) \neq \emptyset$, a {\em daemon} selects a
non-empty set $\mathcal X \subseteq \mbox{\it Enabled}(\gamma)$; then
every process of $\mathcal X$ {\em atomically} executes one of its
enabled rule, leading to a new configuration $\gamma^\prime$, and so
on.  The transition from $\gamma$ to $\gamma^\prime$ is called a {\em
  step} (of ${\mathcal A}$). The possible steps induce a binary
relation over $\mathcal C_{\mathcal A}$, denoted by $\mapsto_{\mathcal
  A}$ (or, simply $\mapsto$, when it is clear from the context).  An
{\em execution\/} of $\mathcal A$ is a maximal sequence of its
configurations $e=\gamma_0\gamma_1\ldots \gamma_i\ldots$ such that
$\gamma_{i-1}\mapsto\gamma_i$ for all $i>0$. The term ``maximal''
means that the execution is either infinite, or ends at a {\em
  terminal\/} configuration in which no rule of ${\mathcal A}$ is
enabled at any process.  We denote by $\mathcal{E}_{\mathcal A}$ (or,
simply $\mathcal{E}$, when it is clear from the context) the set of
all possible executions of ${\mathcal A}$.  The set of all executions
starting from a particular configuration $\gamma$ is denoted
$\mathcal{E}(\gamma)$. Similarly, $\mathcal{E}(\mathcal S)$ is the set
of execution whose the initial configuration belongs to $\mathcal S
\subseteq \mathcal C$.

As previously stated, each step from a configuration to another is
driven by a daemon.  In this paper we assume the daemon is {\em
  distributed} and {\em unfair}. ``Distributed'' means that while the
configuration is not terminal, the daemon should select at least one
enabled process, maybe more. ``Unfair'' means that there is no fairness
constraint, {\em i.e.}, the daemon might never select an enabled
process unless it is the only enabled process.

We say that a process $p$ is \emph{neutralized\/} during the step
$\gamma_i \mapsto \gamma_{i+1}$ if $p$ is enabled at $\gamma_i$ and
not enabled at $\gamma_{i+1}$, but does not execute any rule between
these two configurations.  An enabled process is neutralized
if at least one neighbor of $p$ changes its
state between $\gamma_i$ and $\gamma_{i+1}$, and this change makes the
guards of all rules of $p$ {\em false}.
To evaluate time complexity, we use the notion of \emph{round}.  This
notion captures the execution rate of the slowest process in any
execution.  The first round of an execution $e$, noted $e^{\prime}$,
is the minimal prefix of $e$ in which every process that is enabled in
the initial configuration either executes an action or becomes
neutralized.  Let $e^{\prime \prime}$ be the suffix of $e$ starting
from the last configuration of $e^{\prime}$.  The second round of $e$
is the first round of $e^{\prime \prime}$, and so forth.

\subsection{Self-Stabilization and Silence}

We are interesting in algorithms which converge from an arbitrary
configuration to a configuration where output variables define a
specific data structure, namely a BFS spanning tree. Hence,
we define a specification as a predicate $\mathds{SP}$ on $\mathcal C$
which is {\em true} if and only if the outputs define the expected
data structure.

{\em Silent Self-stabilization} is a particular form of
self-stabilization defined by Dolev {\em et al}~\cite{DolevGS96} as
follows. A distributed algorithm $\mathcal A$ is {\em silent
  self-stabilizing w.r.t. specification $\mathds{SP}$} if following
two conditions holds:
\begin{description}
\item[Termination:] all executions of $\mathcal A$ are finite; and
\item[Partial Correctness:] all terminal configurations of $\mathcal A$ satisfy $\mathds{SP}$.
\end{description}

In this context, the {\em
  stabilization time} is the maximum time (in steps or rounds) to
reach a terminal configuration starting from any configuration.
 
\section{Three Self-Stabilizing BFS Constructions}\label{sect:algo}

Below we give three self-stabilizing BFS constructions: $\algoU$,
$\algoB(D)$, and $\algoHC(D)$.\footnote{$\mathcal U$, $\mathcal B$,
  $\mathcal{HC}$ respectively stand for {\em unbounded}, {\em
    bounded}, and {\em Huang-Chen}.} In all these variants, the local
program of the $\Root$ just consists in the following constant:
$d_{\Root} = 0$.

Each non-root process $p$ maintains two variables: $d_p$ and
$par_p$. The domain of $par_p$ is $\neig_p$, the set of $p$'s
neighbors. The domain of $d_p$ differs depending on the version:
\begin{itemize}
\item $d_p$ is an unbounded positive integer in the first version,
  $\algoU$.
\item $d_p \in [1..D]$ in the two other algorithms, $\algoB(D)$ and
  $\algoHC(D)$. The correctness of both $\algoB(D)$ and $\algoHC(D)$
  will be established for any $D \geq \diam$.
\end{itemize}
The three algorithms use the following three macros:
\begin{itemize}
\item $\mind(p) = \min \{d_q\ |\ q \in \neig_p\}$
\item $\Best(p)$  is  any neighbor $q$ such that $d_q = \mind(p)$
\item $\update(p)$: $d_p \gets \mind(p)+1$; $par_p \gets \Best(p)$
\end{itemize}
Moreover, the following two predicates are used in the algorithms:
\begin{itemize}
\item $dOk(p) \equiv (d_p = \mind(p) + 1)$
\item $parOk(p) \equiv  (d_p = d_{par_p} +1)$
\end{itemize}
The two variables are maintained using the actions defined below.

\subsection{Actions of $\algoU$}\label{sub:u}

The BFS algorithm in~\cite{DIM93} is designed
for the read/write atomicity model.
The straightforward adaptation of this algorithm
in the locally shared memory model is given below.\\

\begin{tabular}{cclcl}
$U_1$ & $::$ & $\neg dOk(p)$ & $\to$ & $\update(p)$
\\
$U_2$ & $::$ & $dOk(p) \wedge \neg parOk(p)$ & $\to$ & $par_p \gets \Best(p)$;\\
\end{tabular}

\subsection{Actions of $\algoB(D)$}\label{sub:d}

The following algorithm is a variant of $\algoU$, where the
domain of the $d$-variable is now bounded by the input parameter $D$.\\

\begin{tabular}{cclcl}
$B_1$ & $::$ & $\mind(p)<D \wedge \neg dOk(p)$ & $\to$ & $\update(p)$\\

$B_2$ & $::$ & $\mind(p)<D \wedge dOk(p) \wedge \neg parOk(p)$ & $\to$ & $par_p \gets \Best(p)$;\\
$B_3$ & $::$ & $\mind(p)=D \wedge( d_p\neq D)$ & $\to$ & $d_p \gets D$\\
\end{tabular}

\subsection{Actions of $\algoHC(D)$}\label{sub:HC}

Actions given below are essentially the same as those of the BFS
algorithm given in~\cite{HC92}.  Actually, they differ in two points
from the version of~\cite{HC92}.
\begin{itemize}
\item In~\cite{HC92}, $d$-variables are defined in such way that $d_R
  = 1$ and $d_p \geq 2, \forall p \in V \setminus \{\Root\}$. We have
  changed the domain definition of the $d$-variables for sake of
  uniformity. However, this difference on the domain definitions has
  no impact on the behavior of the algorithm.
\item Moreover, to be more general, we have replaced in the code the
  exact value $n$ (the number of processes) by $D$. 
\end{itemize}

\begin{tabular}{cclcl}
$HC_1$ & $::$ & $\neg parOk(p) \wedge  d_{par_p} < D$ & $\to$ & $d_p \gets d_{par_p}+1$\\
$HC_2$ & $::$ & $d_{par_p} > \mind(p)$ & $\to$ & $\update(p)$\\
\end{tabular}

\bigskip

Notice that $HC_1$ and $HC_2$ are not mutually exclusive, {\em i.e.},
in some configurations, both rules can be enabled at the same process.
For instance, in the initial configuration of Figure \ref{fig:execHC}
(page \pageref{fig:execHC}), rules $HC_1$ and $HC_2$ are enabled at
process $a$. In this case, if $a$ is selected by the daemon, the
daemon also chooses which of the two rules is executed.

\section{Correctness}\label{sect:corr}

In this section, we give a general proof which establishes the
self-stabilization of the three algorithms under a distributed unfair
daemon. The proof of correctness consists of the following two main steps:
\begin{description}
\item[Partial correctness (Theorems~\ref{bfs:5} and~\ref{pc})] which
  means that if an execution terminates, then the output of the
  terminal configuration is correct.  (In our context, the output of
  the terminal configuration define a BFS spanning tree rooted at
  \Root.)
\item[Termination (Theorem~\ref{term})] which means that every
  possible execution
%(under a distributed unfair daemon) 
  is finite in terms of steps.
\end{description}

\subsection{Partial Correctness}

Below we establish the partial correctness of $\algoU$, $\algoB(D)$,
and $\algoHC(D)$ using three main steps: (1) we define a set of {\em
  legitimate} configurations (Definition~\ref{def:legi}), and show
that (2) a BFS spanning tree is defined in each legitimate
configuration (Theorem~\ref{bfs:5}) and (3) every terminal
configuration of each algorithm is legitimate
(Theorem~\ref{pc}). additionally, we show that every legitimate
configuration is terminal in any of the three algorithms
(Theorem~\ref{legi}).

\subsubsection{Legitimate Configurations}
\begin{definition}\label{def:legi}
  A configuration is {\em legitimate} if and only if for every process
  $p$, we have $d_p = \|p,\Root\|$ and if $p \neq \Root$, then $d_{p}
  = d_{par_p}+1$.
\end{definition}

Let $T_\gamma = (V,E_T)$, where $E_{T_\gamma} = \{\{p,q\} \in E\ |\ q
\neq \Root \wedge par_q = p$ in $\gamma\}$.
 
\begin{theorem}\label{bfs:5}
 $T_\gamma$ is a BFS spanning tree in every legitimate
  configuration $\gamma$.
\end{theorem}
\begin{proof}
Let $\gamma$ be any legitimate configuration.  We first demonstrate
that $T_\gamma$ is a spanning tree by showing the following two
claims:
\begin{description}
\item[$T_\gamma$ is acyclic:] Assume, by contradiction, that
  $T_\gamma$ contains a cycle $p_0,\ldots, p_k, p_0$. By
  definition, $\Root$ is not involved into the cycle. Assume,
  without loss of generality, that for all $i>0$, $par_{p_{i}} =
  p_{i-1}$ et $par_{p_0} = p_{k}$ in $\gamma$. From
  Definition~\ref{def:legi}, in $\gamma$ we have $d_{p_i}>d_{p_{i-1}}$
  and, by transitivity, $d_{p_k} > d_{p_0}$. Now, as $par_{p_0} =
  p_{k}$, we have $d_{p_0} > d_{p_k}$ by Definition~\ref{def:legi}, a
  contradiction.

\item[$|E_{T_\gamma}| = n-1$:] First, by definition we have
  $|E_{T_\gamma}| \leq n-1$. Now, if $|E_{T_\gamma}| < n-1$, then
  there is at least one edge $\{p,q\}$ such that $par_p = q$ and $par_q
  = p$.  This contradicts the fact that $T_\gamma$ is acyclic. Hence,
  $|E_{T_\gamma}| = n-1$.
\end{description}
We now show that $T_\gamma$ is breadth-first.  Let $p_0=\Root, \ldots,
p$ the unique path from $\Root$ to $p$ in the tree.  By definition,
for all $i>0$, $par_{p_{i}} = p_{i-1}$.  By Definition~\ref{def:legi},
we have $d_{p_i}=d_{p_{i-1}}+1$ and, by transitivity,
$d_{p_k}=d_{p_{0}}+k$. Moreover, $d_{p_0} = d_\Root = 0$. So,
$d_{p_k}$ is the length $k$ of the path from $\Root$ to $p_k$. Now, by
Definition~\ref{def:legi}, $d_{p_k}$ is also equal to
$\|\Root,p_k\|$. Hence, the length $k$ of the path from $\Root$ to
$p_k$ is equal to the distance from $\Root$ to $p_k$ in $G$.
\end{proof}

\subsubsection{Legitimacy of Terminal Configurations}
Let $D\geq \diam$.  Let $\mathcal{TC}_{\algoU}$,
$\mathcal{TC}_{\algoB(D)}$, and $\mathcal{TC}_{\algoHC(D)}$ be the set of
terminal configurations of $\algoU$, $\algoB(D)$, and $\algoHC(D)$,
respectively.

\begin{lemma}
Let $\gamma$ be a configuration of
$\mathcal{TC}_{\algoU}$. Let $X$ be the largest distance value in $\gamma$.
$\gamma$ is a
configuration of $\mathcal{TC}_{\algoB(X)}$. 
\end{lemma}
\begin{proof}
  Every process $p$ satisfies $dOk(p) \wedge parOk(p)$ in $\gamma$.
  So, the rule $B_1$ and $B_2$ of $\mathcal{TC}_{\algoB(X)}$ are
  disabled at $p$ in $\gamma$. Moreover, by definition, $d_p \leq X$
  and $parOk(p)$ in $\gamma$ implies that $d_{par_p} < X$ in
  $\gamma$. So, $\mind(p) < X$ in $\gamma$ and, consequently, $B_3$ of
  $\mathcal{TC}_{\algoB(X)}$ is disabled at $p$ in $\gamma$.
\end{proof}

\begin{lemma}
Let $D \geq \diam$ and $\gamma$ be a configuration. If $\gamma$ be a terminal configuration of $\algoB(D)$, then 
$\gamma$ is a terminal configuration of $\algoHC(D)$. 
\end{lemma}
\begin{proof}
  We establish this lemma by showing its contrapositive. Let $\gamma'$
  be a non-terminal configuration of $\algoHC(D)$. There is a process
  $p$ such that $(\neg parOk(p) \wedge d_{par_p} < D) \vee (d_{par_p} >
  \mind(p))$ holds in $\gamma'$.

  Assume first that $\neg parOk(p) \wedge d_{par_p} < D$ holds in
  $\gamma'$. Then, $d_{par_p} < D$ implies that $\mind(p)<D$
  holds. Thus either $B_1$ or $B_2$ is enabled in $\gamma'$ because
  $\neg parOk(p)$ holds.

  Assume then that $d_{par_p} > \mind(p)$.  Then, as $d_{par_p} \leq
  D$ (by definition), we have $\mind(p) < D$. So, if $\neg dOk(p)$,
  then $B_1$ is enabled in $\gamma'$. Otherwise, we have $d_p =
  \mind(p)+1$. So, $d_p \neq d_{par_p}+1$ and, consequently, $B_2$ is
  enabled in $\gamma'$.

  Hence, every non-terminal configuration of $\algoHC(D)$ is a
  non-terminal configuration of $\algoB(D)$.
\end {proof}

From the two previous lemmas, we have:

\begin{corollary}\label{coro:term}
$\mathcal{TC}_{\algoU} \subseteq \bigcup_{i=\diam}^\infty
  \mathcal{TC}_{\algoB(i)} \subseteq \bigcup_{i=\diam}^\infty
  \mathcal{TC}_{\algoHC(i)}$.
\end{corollary} 

From the previous corollary, we know that it is sufficient to show that
any configuration of $\bigcup_{i=\diam}^\infty \mathcal{TC}_{\algoHC(i)}$
is legitimate to establish that any configuration of
$\mathcal{TC}_{\algoU}$, $\bigcup_{i=\diam}^\infty \mathcal{TC}_{\algoB(i)}$, 
and $\bigcup_{i=\diam}^\infty  \mathcal{TC}_{\algoHC(i)}$ are legitimate.

\begin{lemma}\label{lem:leq}
Let $D \geq \diam$.
In any configuration of $\mathcal{TC}_{\algoHC(D)}$ 
we have $d_p \leq \|p,\Root\|$ for every
process $p$.
\end{lemma}
\begin{proof}
  Let $\gamma$ be any terminal configuration of $\algoHC(D)$. Assume,
  by the contradiction, that there is a process $p$ such that $d_p >
  \|p,\Root\|$ in $\gamma$.  Without loss of generality, assume $p$ is
  one of the closest processes from $\Root$ such that $d_p >
  \|p,\Root\|$ in $\gamma$. By definition, $\|p,\Root\| \leq \diam$,
  and $p$ has at least a neighbor, $q$, such that $\|q,\Root\| =
  \|p,\Root\|-1 < \diam$.  By hypothesis, $d_q \leq \|q,\Root\| =
  \|p,\Root\|-1$ in $\gamma$.  So, we have $\mind(p) \leq \|p,\Root\|-1 <
  \diam \leq D$ in $\gamma$. Then, by definition, $d_{par_p} \geq
  \mind(p)$ in $\gamma$. Now, $HC_2$ is disabled at $p$ in
  $\gamma$. So, $d_{par_p} = \mind(p)$ in $\gamma$. Consequently,
  $d_{par_p} \leq \|p,\Root\|-1 < D$ in $\gamma$. Now, $HC_1$ is
  disabled at $p$ in $\gamma$. So, $d_p = d_{par_p} + 1 \leq
  \|p,\Root\|$ in $\gamma$, a contradiction.
\end{proof}

\begin{corollary}
\label{cor:leq}
Let $\gamma$ be a configuration of $\mathcal{TC}_{\algoHC(D)}$ where $
D \geq \diam$.  Every process $p \neq \Root$ satisfies $d_p =
d_{par_p}+1$ and $d_{par_p} = \mind(p)$ in $\gamma$.
\end{corollary}
\begin{proof}
  Let $p\neq \Root$ be process. 
By definition, $d_{par_p} \geq \mind(p)$ in $\gamma$. 
$HC_2$ is disabled at $p$ in $\gamma$. So,
  $d_{par_p} = \mind(p)$ in $\gamma$.  $p$ has a neighbor
  such that $\|q,\Root\| \leq \diam -1$.  So, we have $d_q \leq \diam-1$ in
  $\gamma$ (by Lemma \ref{lem:leq}) and, consequently, $\mind(p) < D$
  in $\gamma$. So, $d_{par_p} < D$ in $\gamma$. As $HC_1$ is
  disabled at $p$ in $\gamma$, we have also $d_p = d_{par_p}+1$.
\end{proof}

\begin{lemma}\label{lem:eq}
Let $D \geq \diam$.
Let $\gamma$ be a configuration of $\mathcal{TC}_{\algoHC(D)}$. 
$d_p = \|p,\Root\|$ holds for every process $p$, in $\gamma$.
\end{lemma}
\begin{proof}
  We have $d_\Root = 0$ (by definition) and $d_q = d_{par_q}+1$ for
  every process $q \neq \Root$ in $\gamma$ (corollary \ref{cor:leq}).
  So, similarly to the proof of Theorem \ref{bfs:5}, we can establish
  that $T_\gamma$ is spanning tree and for every process $p$, $d_p$ is
  the length of the path in $T_\gamma$ from $\Root$ to $p$. So, we
  have $d_p \geq \|p,\Root\|$ in $\gamma$.  Now, $d_p \leq
  \|p,\Root\|$ in $\gamma$, by Lemma~\ref{lem:leq}. So, we conclude
  that $d_p = \|p,\Root\|$ in $\gamma$, for every process $p$.

% Let $\gamma$ be any terminal configuration of $\algoHC(D)$.  First,
%  $d_R = 0 = \|\Root,\Root\|$.  Assume then, by the contradiction,
%  that there is a process $p\neq \Root$ such that $d_p \neq
%  \|p,\Root\|$ in $\gamma$.  By Lemma~\ref{lem:leq}, $d_p <
 % \|p,\Root\|$ in $\gamma$.  Without loss of generality, assume that
 % $p$ is the process with the minimum value of $d_p$ satisfying $d_p <
%  \|p,\Root\|$ in $\gamma$.  According to Corollary \ref{cor:leq}, we
%  have $d_p = d_{par_p}+1$.  So, by hypothesis, we have $d_{par_p} =
%  \|par_p, \Root\|$.  We conclude that $d_p = \|par_p,\Root\|+1 \geq
%  \|p,\Root\|$ in $\gamma$, a contradiction.
\end{proof}

From Corollaries~\ref{coro:term},\ref{cor:leq}, and
Lemma~\ref{lem:eq}, we can deduce the following theorem:

\begin{theorem}\label{pc}
  Let $D \geq \diam$. Every terminal configuration of 
  $\algoU$, $\algoB(D)$, or $\algoHC(D)$
is $\gamma$ is legitimate.
\end{theorem}

\subsubsection{Legitimate Configurations are Terminal}
\begin{theorem}\label{legi}
  Let $D \geq \diam$.  Every legitimate configuration is a terminal
  configuration of $\algoU$, $\algoB(D)$, and $\algoHC(D)$,
  respectively.
\end{theorem}

\begin{proof} 
  Let $\gamma$ be a legitimate configuration. First, for every process
  $p$, $d_p = \|p,\Root\| \leq \diam$. So, $\gamma$ is a possible
  configuration of $\algoU$, $\algoB(D)$, and $\algoHC(D)$,
  respectively.

  Let $p$ be a non-root process.  We have $\mind(p) = \min \{d_q\ |\ q
  \in \neig_p\} = \min \{\|q,\Root\|\ |\ q \in \neig_p\} = \|p,\Root\| -
  1 < \diam \leq D$ in $\gamma$.  So, $B_3$ is disabled at every
  non-root process in $\gamma$.

  Moreover, $\mind(p)+1 = \|p,\Root\| = d_p$ in $\gamma$. So, $U_1$
  and $B_1$ are disabled at every non-root process in $\gamma$.

  By definition, $parOk(p)$ holds in $\gamma$. So, $U_2$, $B_2$, and
  $HC_1$ are disabled at non-root process in $\gamma$.

  Finally, $parOk(p)$ implies that $d_{par_p} = d_p-1 = \|p,\Root\| -
  1 =\mind(p)$ in $\gamma$. Hence, $HC_2$ is disabled at every
  non-root process in $\gamma$.
\end{proof}

\subsection{Termination}

In this subsection, we will establish that, under a distributed unfair
daemon, all executions of $\algoU$, $\algoB(D)$, and $\algoHC(D)$ are
finite. 

The lemma given below establish that we only need to prove that the
number of $d$-variable updates is finite in any execution $e$ of $\algoU$,
$\algoB(D)$, or $\algoHC(D)$ to establish that $e$ is
finite.

\begin{lemma}\label{lem:par}
  Let $e = \gamma_0, \ldots \gamma_i, \gamma_{i+1}, \ldots$ be any
  execution of $\algoU$, $\algoB(D)$, or $\algoHC(D)$.
  If for every process $p$, $d_p$ is modified only a finite number
  of time along $e$, then $e$ is finite.
\end{lemma}
\begin{proof}
  Assume that every process $p$, $d_p$ is modified only a finite
  number of time along $e$. Then, there exists $i \geq 0$ such that no
  $d$-variable is modified in the suffix
  $e'=\gamma_i\gamma_{i+1}\ldots$ of $e$. By definition of the three
  algorithms, only $par$-variables can be modified along $e'$.  So the
  rules $U_1$ for $\algoU$, $B_1$ and $B_3$ for $\algoB(D)$, and
  $HC_1$ for $\algoHC(D)$ are not executed along $e'$.  Now, by
  definition of the algorithms, in $e'$, we have:
  \begin{itemize}
  \item Once the rule modifying $par_p$ (Rule $U_2$, $B_2$, or $HC_2$)
    is disabled, it remains disabled forever by $p$, because the
    values of $d$-variables are constant (in particular, those of $p$
    and its neighbors).
  \item The rule modifying $par_p$ (Rule $U_2$, $B_2$, or $HC_2$)
    becomes disabled immediately after $p$'s execution.
\end{itemize}
Consequently, each process takes at most one step along $e'$; we
conclude that the execution $e$ is finite.
\end{proof}

\begin{notation}
  For every configuration $\gamma$, for any integer $k \geq 0$, we
  denote by $Set\_d_k(\gamma)$ the set of processes $p$ such that $d_p
  = k$ in $\gamma$.
\end {notation}

\begin{remark}\label{rem:R}
  In every configuration $\gamma$, $Set\_d_0(\gamma) = \{R\}$, and
  $Set\_d_\ell(\gamma) \bigcap Set\_d_{k}(\gamma) = \emptyset$ for
  every $0\leq \ell <k$.  $V = \bigcup_{i=0}^{\infty}
  Set\_d_{i}(\gamma)$.
\end{remark}

The following lemma establishes that for every execution $e$ of
$\algoU$, $\algoB(D)$, or $\algoHC(D)$, there is a upper bound $kb$
on the values taken in $e$ by the $d$-variables of all processes.

\begin{lemma}\label{lem:2}
  Let $e = \gamma_0, \ldots \gamma_i, \gamma_{i+1}, \ldots$ be any
  execution of $\algoU$, $\algoB(D)$, or $\algoHC(D)$.  $\exists kb
  \geq 0\ |\ \forall j\geq 0, \forall \ell > kb$ we have
  $Set\_d_\ell(\gamma_j) = \emptyset$.
\end{lemma}
\begin{proof}
  By definition, the lemma is established by letting $kb = D$ if $e$
  is an execution of $\algoB(D)$ or $\algoHC(D)$.

  Consider now the case where $e$ is an execution of $\algoU$.
  $\UNBOUNDED$ is the set of processes that have no upper bound on their
  distance value in $e$, formally, $\UNBOUNDED$ $= \{p \in V$ $|$
  $\forall k \geq 0, \exists j \geq 0, \exists \ell > k\ |\ p \in
  Set\_d_\ell(\gamma_j)\}$

  Assume, by the contradiction, that $\UNBOUNDED$ is not empty.  The
  set $V \setminus \UNBOUNDED$ is not empty, by Remark~\ref{rem:R}. As
  the network is connected, there are two neighboring processes $p$
  and $q$ such that $p \in \UNBOUNDED$ and $q \in V \setminus
  \UNBOUNDED$.  By definition, $\exists x \geq 0$ such that $d_q \leq
  x$ in all configurations of $e$. Consequently, $\mind(p) \leq x$ in
  all configurations of $e$. So, we have $d_p \leq \max \{x+1,y\}$ in
  all configurations of $e$, where $y$ be the initial value of $d_p$
  (according to the rule $U_1$).  Consequently, $p \notin \UNBOUNDED$,
  a contradiction.

  Hence, $\UNBOUNDED$ is empty. Let $ub_p$ the upper bound on the
  distance values taken by the process $p$ in $e$.  The lemma holds
  for $kb = \max_{p \in V}\{ub_p\}$.
\end{proof}

Below, we show that, for every $k \geq 0$, for every execution $e$ of
$\algoU$, $\algoB(D)$, or $\algoHC(D)$, if there is a suffix $e'$
of $e$ where every $d$-variable whose value is less than $k$ is
constant, then there is a suffix $e''$ of $e'$ where no process
switches its $d$-variable from any non-$k$ value to $k$.

\begin{lemma}\label{lem:1}
  Let $e = \gamma_0, \ldots \gamma_i, \gamma_{i+1}, \ldots$ be an
  execution of $\algoU$, $\algoB(D)$, or $\algoHC(D)$. Let $k > 0$.
  If $\exists i_{k}\ |\ \forall j \geq i_k, \forall \ell \in [0..k-1]$
  we have $Set\_d_\ell(\gamma_j) = Set\_d_\ell(\gamma_{i_{k}})$, then
  $\exists \ell \geq i_k\ |\ \forall j\geq \ell$ we have
  $Set\_d_k(\gamma_{j+1}) \subseteq Set\_d_k(\gamma_{j})$.
\end{lemma}
\begin{proof}
  Let $\gamma_j \mapsto \gamma_{j+1}$ be any step in the suffix of $e$
  starting in $\gamma_{i_k}$ where $Set\_d_k(\gamma_{j+1}) \nsubseteq
  Set\_d_k(\gamma_{j})$.  There is at least a process $p \neq \Root$
  such that $p \notin Set\_d_k(\gamma_j) \wedge p \in
  Set\_d_k(\gamma_{j+1})$.  In $\gamma_j$, we have $d_p > k$, as
  $Set\_d_\ell(\gamma_j) = Set\_d_\ell(\gamma_{j+1})$ $\forall \ell
  \in [0..k-1]$ and $p \notin Set\_d_k(\gamma_j)$.  Moreover, $p$
  executes a rule in $\gamma_j \mapsto \gamma_{j+1}$.  In the
  following, we prove that $p$ will no more change its $d$-variable in
  $e$ after this step.

\begin{itemize}
\item Consider first $\algoU$.  We have $\mind(p) = k-1$ in
  $\gamma_j$. Moreover, by hypothesis, $\mind(p) = k-1$ forever from
  $\gamma_{i_k}$ (so, in particular from $\gamma_j$).  So, $p$ will no
  more change its distance value after $\gamma_{j+1}$.

\item Consider $\algoB(D)$.  If $k=D$, then $d_p$ should be greater
  than $D$ in $\gamma_j$, a contradiction.  So, $k < D$ and $d_p > k$
  in $\gamma_j$.  So $p$ executes $B_1$ in $\gamma_j \mapsto
  \gamma_{j+1}$, and similarly to the previous case, $\gamma_j \mapsto
  \gamma_{j+1}$ is the only step in the suffix of $e$ starting in
  $\gamma_{i_k}$ where $p$ sets $d_{p}$ to $k$.

\item Finally, consider $\algoHC(D)$.
We have to study the two following cases:
\begin{itemize}
\item Assume that $p$ executes $HC_2$ to set $par_{p}$ to $q$ in
  $\gamma_j \mapsto \gamma_{j+1}$.  By definition, $\mind(p) = k-1 =
  d_{q}$ holds in $\gamma_{i_k}$ and all subsequent configurations.
  So, $p$ is disabled forever from $\gamma_{j+1}$.

\item Assume that $p$ executes $HC_1$ to set $par_{p}$ to $q$ in
  $\gamma_j \mapsto \gamma_{j+1}$ By definition, $\mind(p) \leq k-1 =
  d_{q} < D$ in $\gamma_{i_k}$ and all subsequent configurations. So,
  until $p$ next action, we have $parOk(p)$ and $d_p = k$.  So, $p$
  next action is necessarily $HC_2$ to set $d_{p}$ to a value smaller
  than $k$, a contradiction.  So, $p$ cannot execute any rule in the
  suffix starting from $\gamma_{j+1}$.
\end{itemize}
\end{itemize}
Hence, in the suffix of $e$ starting in $\gamma_{i_k}$, there is at
most $n$ steps $\gamma_j \mapsto \gamma_{j+1}$ where
$Set\_d_k(\gamma_{j+1}) \nsubseteq Set\_d_k(\gamma_{j})$.
\end{proof}

Below, we show that, for every $k \geq 0$, for every execution $e$ of
$\algoU$, $\algoB(D)$, or $\algoHC(D)$, if eventually every
$d$-variable whose value is less than $k$ becomes constant, then
eventually every $d$-variable whose value is $k$ becomes constant.

\begin{lemma}\label{lem:1b}
  Let $e = \gamma_0, \ldots \gamma_i, \gamma_{i+1}, \ldots$ be an
  execution of $\algoU$, $\algoB(D)$, or $\algoHC(D)$. Let $k > 0$.
  If $\exists i_{k}\ |\ \forall j \geq i_k, \forall \ell \in [0..k-1]$
  we have $Set\_d_\ell(\gamma_j) = Set\_d_\ell(\gamma_{i_{k}})$, then
 $\exists \ell \geq i_k\ |\ \forall j\geq \ell$ we have
  $Set\_d_k(\gamma_{j+1}) = Set\_d_k(\gamma_{j})$.
\end{lemma}
\begin{proof}
  By lemma \ref{lem:1}, there exists a suffix $e'$ of $e$ starting in
  $\gamma_x$, such that $\forall j\geq x$, we have
  $Set\_d_k(\gamma_{j+1}) \subseteq Set\_d_k(\gamma_{j})$.  During
  $e'$, there is at most $|Set\_d_k(\gamma_{i_x})|$ steps $\gamma_j
  \mapsto \gamma_{j+1}$ where $Set\_d_k(\gamma_{j+1}) \neq
  Set\_d_k(\gamma_{j})$.
\end{proof}

From Remark \ref{rem:R}, Lemmas \ref{lem:2} and \ref{lem:1b},
 we can deduce the following
corollary.

\begin{corollary}\label{coro:d}
  For every process $p$, $d_p$ can be modified only a finite number of
  time in $e$.
\end{corollary}

By Lemma~\ref{lem:par},
Corollary~\ref{coro:d}, follows:

\begin{theorem}\label{term}
  Under a distributed unfair daemon, all executions of $\algoU$,
  $\algoB(D)$, and $\algoHC(D)$ are finite.
\end{theorem}

\section{Stabilization Time in Rounds}\label{sect:round}

In this section, we study the stabilization time in rounds of the
three algorithms presented in Section~\ref{sect:algo}. Throughout this
section we will use the notion of {\em attractor} defined below.  

Let $\mathcal A$ be a distributed algorithm. Let $\mathcal C_1$ and
$\mathcal C_2$ be two subsets of $\mathcal C$, the set of all possible
configurations of $\mathcal A$.  $\mathcal C_2$ is an {\em attractor}
for $\mathcal C_1$ (under $\mathcal A$) if the following conditions
hold:
\begin{description}
\item[Convergence:] $\forall e = \gamma_0, \gamma_1, \ldots \in
  \mathcal{E}({\mathcal C_1}), \exists i \geq 0\ |\ \gamma_i \in \mathcal
  C_2$.
\item[Closure:] $\forall e = \gamma_0, \gamma_1, \ldots \in
  \mathcal{E}({\mathcal C_2}), \forall i \geq 0\ |\ \gamma_i \in \mathcal
  C_2$.
\end{description}
The following predicate is useful to establish a sequence of
attractors.

$$\correctb(p,i) \equiv (\|\Root,p\| \leq i \Rightarrow (d_p = \|\Root,p\| = d_{par_p}+1))$$

For every $i \geq 0$, $\LA(i)$ is the set of configurations, where
every process $p \neq \Root$ satisfies $\correctb(p,i)$.

In any configuration of $\LA(\diam)$, every process $p \neq \Root$ satisfies
$(d_p = \|\Root,p\| = d_{par_p}+1)$, moreover $d_r=0$, by definition.
So, all configurations of $\LA(\diam)$ are legitimate.  Furthermore,
every legitimate configuration is terminal in all of the three
algorithms (Theorem~\ref{legi}). Hence, the stabilization time of any
of the three algorithms is bounded by the maximum number of rounds
it requires to reach any configuration of $\LA(\diam)$ starting from
any arbitrary configuration.

\subsection{Lower Bound in $\Omega(D)$ Rounds for $\algoHC(D)$}
\label{sub:sec-cvHC}

We first show that the stabilization time in rounds of $\algoHC(D)$
actually depends on the size of the domain of the
$d$-variables. Hence, we can conclude that $\algoHC(n)$, {\em i.e.},
the algorithm proposed in~\cite{HC92}, stabilizes in $\Omega(n)$
rounds, where $n$ is the number of processes.

Let $k \geq 1$. We now exhibit a possible execution of $\algoHC(2k)$
which stabilizes in $k+1$ rounds in the 3-nodes graph given in
Figure~\ref{fig:execHC} (its diameter is 2). Notice that this
execution requires $2k$ steps.

\begin{figure}[h]
\begin{centering}
\scalebox{0.7}{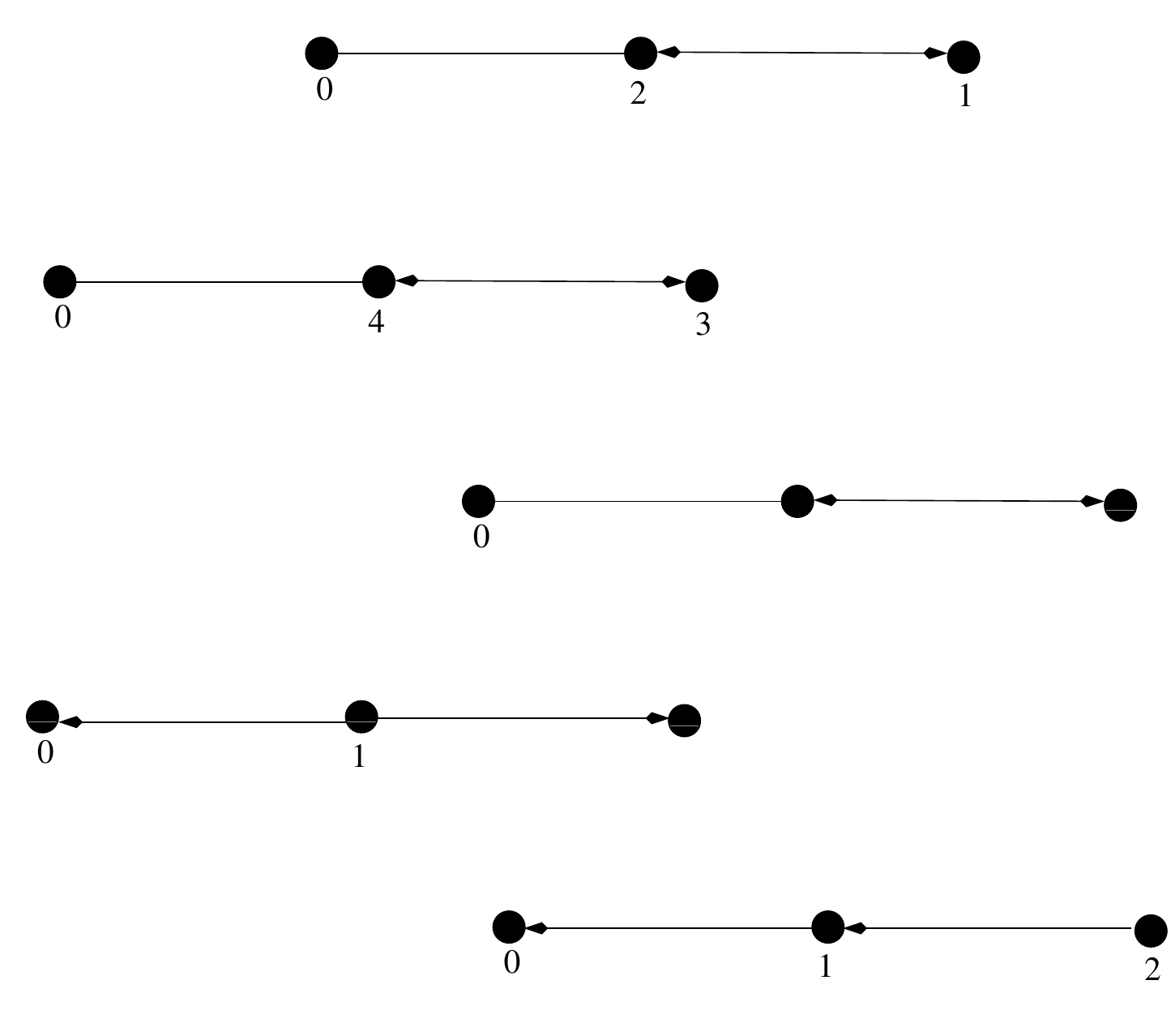}
\caption{Execution of $\algoHC(2k)$ which converges in $k+1$ rounds}
\label{fig:execHC}
\end{centering}
\end{figure}

\begin{itemize} 
\item At the beginning of the $i$th round with $i \in [1,k-1]$,
  processes $a$ and $b$ are enabled.  In the first step of these
  rounds, $b$ executes $HC_1$. During the second step (the last step
  of these rounds) the node $a$ executes $HC_1$.
\item At the beginning of the $k$th round, only process $a$ is
  enabled.  During the only step of this round, $a$ executes $HC_2$
  and gets its terminal state.
\item At the beginning of the $k+1$th round, only process $b$ is
  enabled.  During the only step of this round, process $b$ executes
  and gets its terminal state.
\end{itemize}
This example can be generalized to any number of processes $n \geq
3$. Just construct a network $G$ of $n$ processes by adding $n-3$
processes to the network given in Figure \ref{fig:execHC}; those $n-3$
processes being only neighbors of $\Root$. Since the state of $\Root$
is constant, these $n-3$ processes have no impact on the behavior of
$a$ and $b$. Hence, the previous execution is a possible execution
prefix in $G$ which contains $\Omega(D)$ rounds.

Hence, the stabilization time of $\algoHC(D)$ is $\Omega(D)$ rounds.

The lower bound on the stabilization time is mainly due to the fact
that rules $HC_1$ and $HC_2$ are not mutually exclusive. Hence, when
both are enabled at the same process $p$, the daemon may choose to
activate any of them. Our lower bound is then established when the
daemon makes priority on $HC_1$.

In the following subsection, we show that this lower bound can be
easily circumvented to obtain the stabilization time in
$\Theta(\diam)$ rounds.

\subsection{Fast Implementation of
  $\algoHC(D)$}\label{sub:FHC}

\subsubsection{Algorithm~$\algoFHC(D)$}

Below, we propose a variant of $\algoHC(D)$ where we have modified $HC_1$
into $FHC_1$, so that $FHC_1$ and $HC_2$ are now mutually
exclusive. The modification of $HC_1$ into $FHC_1$ gives it priority on
$HC_2$. In the following, this variant will be denoted by $\algoFHC(D)$ and called {\em fast implementation of $\algoHC(D)$}. 

\begin{tabular}{cclcl}
$FHC_1$ & $::$ & $\neg parOk(p) \wedge  d_{par_p} < D
\wedge d_{par_p} = \mind(p)$ & $\to$ & $d_p \gets d_{par_p}+1$\\
$HC_2$ & $::$ & $d_{par_p} > \mind(p) $ & $\to$ & $update(p)$\\
\end{tabular}\\

The lemma given below show the close relationship between
$\algoFHC(D)$ and $\algoB(D)$.

\begin{lemma}
\label{lem:fast-impl1}
  If $\gamma \mapsto \gamma'$ is a step of $\algoFHC(D)$ containing
  execution of rules $HC2$ only, then $\gamma \mapsto \gamma'$ is a
  possible step of $\algoB(D)$.
\end{lemma}
\begin{proof}
Let $\gamma \mapsto \gamma'$ be any step of $\algoFHC(D)$ containing
execution of rules $HC2$ only.  Consider any process $p$ that moves
during $\gamma \mapsto \gamma'$.  So, $p$ performs $HC_2$ during
$\gamma \mapsto \gamma'$ and we have $\mind(p) < d_{par_p} \leq D$ in
$\gamma$.

If $\neg dOk(p)$ is true in $\gamma$, then $B_1$ is enabled at $p$ in
$\gamma$.  Now, the action part of $B_1$ and $HC_2$ are identical.

Conversely, assume that $dOk(p)$ is true in $\gamma$.  During $\gamma
\mapsto \gamma'$, $p$ does not modify $d_p$, however $par_p$ is set to
$\Best(p)$. Now, $\neg parOk(p)$ is true in $\gamma$, so $B_2$ is
enabled while $B_1$ is not.  In this case, the action part of $B_2$
has the same effect as the action part of $HC_2$.

Hence, in any case, $\gamma \mapsto \gamma'$ is a possible step of
  $\algoB(D)$.
\end{proof}

\subsubsection{Upper Bound on Stabilization Time in Rounds of $\algoFHC(D)$}

The $\correctx(p,i)$ and $\correctc(p,i)$ predicate defined below are used 
to establish a sequence of
$\diam+1$ attractors under $\algoFHC(D)$ (with $D
\geq \diam$) ending in the set of the terminal configurations.

$$\correctx(p,i) \equiv (\|\Root,p\| > i \Rightarrow (d_p > i \vee (d_p=i \wedge (\exists q  \in \neig_p\ |\ d_q \leq i+1))))$$

$\correctx(p,i)$ means that if a process $p$ is at distance larger than $i$ from $\Root$, then either $d_p$ should be also larger than $i$, or $d_p$ should be equal to $i$ and a neighbor of $p$ should have its distance to $\Root$ smaller than or equal to $i+1$.

$$\correctc(p,i) \equiv (\|\Root,p\| \leq i \Rightarrow  d_p = \|\Root,p\|)$$ 

$\correctc(p,i)$ means that if a process $p$ is at least  at distance $i$ from $\Root$, then its distance value should be correct, {\em i.e.}, $d_p$ is equal to its distance to the $\Root$.

Below, we define some useful subsets of configurations.
\begin{itemize}
\item Let $\UBD(i)$ be the set of configurations,
where every process $p \neq \Root$ satisfies $\correctx(p,i)$.
\item Let $\CD(i)$ be the set of configurations,  
where every process $p$ satisfies $\correctc(p,i)$. 
\item Let $\B(i) = \CD(i) \cap \UBD(i)$.
\end{itemize}

Notice that $\B(0)$ is the set of all possible configurations.

\begin{observation}
\label{obs:help}
Let $p$ be an process such that $\|p,\Root\| > i+1$.  So, we have $i+2
\leq \diam$.  Let $\gamma$ be a configuration of $\B(i)$.  By
definition of $\B(i)$, we have $d_p \geq i$ and $d_q \geq i$ for every
$q \in \neig_p$ in every execution from $\gamma$.  Consequently,
$\mind(p) \geq i$ along any execution from $\gamma$.
\end{observation}

\begin{lemma}
 Assume that $\B(i)$ is an attractor under $\algoFHC(D)$  
with $D \geq \diam$ and $0 \leq i < \diam$.
Let $\gamma$ be a configuration of $\B(i)$.
Let $\gamma \mapsto \gamma'$ be a possible step 
where process $p$ moves.
$\correctx(p,i+1)$ holds in $\gamma'$. 
\label{lem:Bi-help}
\end{lemma}
\begin{proof}
  Let $p$ be a process such that $\|p,\Root\| > i+1$ (the other case
  is trivial).
  We have $d_p \geq i+1$ in $\gamma'$ according to observation \ref{obs:help}.\\
  If $d_p > i+1$ in $\gamma'$, then $\correctx(p,i+1)$ holds.
  Otherwise, we have $d_p = i+1$ in $\gamma'$. Let $q \in \neig_p$
  such that $d_q = \mind(p)$ in $\gamma$. We have $d_q = i$ in
  $\gamma$. By definition of $\B(i)$, we have $\mind(p) \leq i+2$ in
  $\gamma'$. Hence, $\correctx(p,i+1)$ holds in $\gamma'$.
\end{proof}

\begin{lemma}
  If $\B(i)$ is an attractor under $\algoFHC(D)$  with $D
\geq \diam$ and $0 \leq i < \diam$, then $\UBD(i+1)$ is an
  attractor under $\algoFHC(D)$ from $\B(i)$
  which is reached within at most one round from $\B(i)$.
\label{lem:Bi-D}
\end{lemma}
\begin{proof} ~

\begin{description}
\item[Closure.]
Let $\gamma \mapsto \gamma'$ be a possible step
from any
  configuration $\gamma$ of $\B(i)$. 
We show that for every process
  $p \neq \Root$, if $\correctx(p,i+1)$ holds in $\gamma$, then
  $\correctx(p,i+1)$ holds in $\gamma'$. 
Assume $\|p,\Root\| > i+1$
  (the other case is trivial). 
Assume that $p$ does not move during the step; otherwise
 $\correctx(p,i+1)$ holds in $\gamma'$ according to Lemma \ref{lem:Bi-help}.
If $d_p > i+1$ in $\gamma$,
$\correctx(p,i+1)$ holds in $\gamma'$.
Assume now that in
  $\gamma$, $d_p = i+1$ and $p$ has a neighbor $q$ such that $d_q \leq
  i+2$.  While $d_p = i+1$, $\mind(q) \leq i+1$, so $d_q \leq i+2$.
  Hence, we can conclude that $\correctx(p,i+1)$ still holds in $\gamma'$.

\item[Convergence.]
 We now show that for every process $p \neq \Root$,
  $\correctx(p,i+1)$ becomes true within at most one round from any
  configuration $\gamma$ of $\B(i)$. 
Assume $\|p,\Root\| > i+1$ (the other case is trivial). 
\begin{itemize}
\item If $d_p > i+1$ in $\gamma$, then $\correctx(p,i+1)$ holds.
\item If $d_p=i$ in $\gamma$, then $p$ is enabled while $d_p = i$
  because $\mind(p) \geq i$ forever from $\gamma$ (Observation
  \ref{obs:help}).  So, $p$ moves during the first round from
  $\gamma$, and we are done, by Lemma~\ref{lem:Bi-help}.
\item Assume that  $d_p = i+1$ in $\gamma$. 
\begin{itemize}
\item If $p$ moves during the first round from $\gamma$, then
  $\correctx(p,i+1)$ holds after the step, by Lemma \ref{lem:Bi-help}.
\item Assume that $p$ does not move during the first round from $\gamma$. 
  \begin{itemize}
\item If a neighbor of $p$, $q$,
 moves during a step of the round, then after this step 
$\correctx(p,i+1)$ holds because $d_q \leq i+2$.
\item Assume that neither $p$ nor its neighbors move during the round.
  So the value of $d_q$ is less than or equal to $i+2$ in $\gamma$,
  for all $q \in \neig_p$. Indeed, if some neighbor of $p$, $q$,
  satisfies $d_q > i+2$ in $\gamma$, then $q$ stay enabled along the
  round from $\gamma$, because of the state of $p$. This contradicts
  the definition of round.  Hence, the value of $d_q$ is less than or
  equal to $i+2$ in $\gamma$, for all $q \in \neig_p$, and
  consequently $\correctx(p,i+1)$ holds in $\gamma$.
\end{itemize}
\end{itemize}
\end{itemize}
\end{description}
\end{proof}

\begin{lemma}
  If $\B(i)$ is an attractor under $\algoFHC(D)$  with $D
\geq \diam$ and $0 \leq i < \diam$, then $\CD(i+1)$ is an
  attractor under $\algoFHC(D)$ from $\B(i)$
  which is reached within at most one round from any configuration of
  $\B(i)$.
\label{lem:CDi-D}
\end{lemma}
\begin{proof}
  Let $p_{i+1}$ be a process at distance $i+1$ of $\Root$.  $p_{i+1}$
  has at least a neighbor $p_i$ such that $\|p_i,\Root\| = i$.  Let
  $\gamma \in \B(i)$. By definition of $ \B(i)$ and as $\B(i)$ is an
  attractor, we can deduce that $d_{p_i} = i$ and $\forall q \in
  \neig_{p_{i+1}}, d_q \geq i$ forever from $\gamma$. So, from
  $\gamma$ $\mind(p_{i+1}) = i$ forever.

 Consequently, if $d_{p_{i+1}} \neq i+1$ in $\gamma$, then $p_{i+1}$
 is enabled to execute $FHC_1$ or $HC_2$ to set $d_{p_{i+1}}$ to
 $i+1$.

 Moreover, if $d_{p_{i+1}} = i+1$ in $\gamma$, then $p_{i+1}$ cannot
 modify $d_{p_{i+1}}$ in any step from $\gamma$.

 Hence, $\CD(i+1)$ is an attractor under $\algoFHC(D)$ from $\B(i)$ which is reached within at most one round
 from any configuration of $\B(i)$.
\end{proof}

From the two previous lemmas, we can deduce the following corollary.

\begin{corollary}\label{coro:D}
  If $\B(i)$ is an attractor under  $\algoFHC(D)$  with $D
\geq \diam$ and $0 \leq i < \diam$, then $\B(i+1)$ is an attractor
  under $\algoFHC(D)$ from $\B(i)$ which is
  reached within at most one round from any configuration of $\B(i)$.
\label{cor:Bi}
\end{corollary}

The previous corollary establishes that after at most $\diam$ rounds,
the distance value in every process is accurate forever. We now show
one additional is necessary to fix the $par$-variables.

\begin{lemma}\label{plus1}
$\LA(\diam)$ is an attractor under $\algoFHC(D)$ (with $D
\geq \diam$) from $\B(\diam)$ which is 
 reached within at most one round from 
any configuration of $\B(\diam)$.
\label{lem:LDiam-HC}
\end{lemma}

\begin{proof}
  In any configuration of $\B(\diam)$, $d_p = \mind(p) +1$ holds
  forever for every process $p$.  So, the distance value of any
  process stays unchanged along any execution of $\algoFHC(D)$ from a
  configuration of $\B(\diam)$.

  Let $\gamma$ be a configuration of $\B(\diam)$ where $d_{par_p} \neq
  \mind(p)$.  The rule $HC_2$ is enabled at $p$ until $p$ executes it.
  After the execution of this rule, we have $d_p = d_{par_p}+1$.  As
  no process changes its distance value in $\B(\diam)$, $p$ is become
  disabled forever.

  Hence, we conclude that $\LA(\diam)$ is an attractor which is reached
  within at most one round from $\B(\diam)$.
\end{proof}

From Corollary~\ref{coro:D} and Lemma~\ref{plus1}, we have the
following theorem:
\begin{theorem}
For every $D \geq \diam$, the stabilization time of
$\algoFHC(D)$ is at most $\diam+1$ rounds.
\end{theorem}

\subsubsection{Lower Bound on Stabilization Time in Rounds of $\algoFHC(D)$}
Below, we show that the upper bound given in the previous theorem is
exact when $D =\diam$: $\forall \diam \geq 2$, there exists an
execution of $\algoHC(\diam)$ in a graph of diameter $\diam$ that
stabilizing in $\diam+1$ rounds.

We consider any graph $G = (V,E)$ of $\diam+2$ nodes of diameter $\diam \geq 1$, where
\begin{itemize}
\item $V = \{p_0=\Root, \ldots, p_{\diam+1}\}$, and
\item $E = \{\{p_i,p_{i+1}\}| i \in [0..\diam]\} \cup \{\{p_{\diam+1},p_{\diam-1}\}\}$.
\end{itemize}
We consider a synchronous execution ({\em i.e.} an execution where the
distributed unfair daemon activates all enabled processes at each
step) which starts from the following initial configuration:
\begin{itemize}
\item $d_\Root = 0$,
\item $\forall i \in [1..\diam-2], par_{p_i} = p_{i-1} \wedge d_{p_i} = \diam$,
\item $par_{p_{\diam-1}} = p_{\diam+1} \wedge d_{p_{\diam-1}} = \diam$, 
\item $par_{p_{\diam}} = p_{\diam+1} \wedge d_{p_{\diam}} = \diam$, and
\item $par_{p_{\diam+1}} = p_{\diam-1} \wedge d_{p_{\diam}} = \diam-1$.
\end{itemize}
An example of initial configuration is given in Figure \ref{fig:execAlgoFHC}.
Notice that in a synchronous execution, every round lasts one step. 
\begin{itemize}
\item At each round $i \in [1..\diam-2]$, $p_i$ executes $FHC_1$ to
  change $d_{p_i}$ to $i$.
\item At the $\diam-1$th round, $p_{\diam - 1}$ executes $HC_2$ to set
  $d_{p_{\diam - 1}}$ to $\diam-1$ and $par_{p_{\diam-1}}$ to
  $p_{\diam - 2}$.
\item At the $\diam$th round, $p_{\diam + 1}$ executes $FHC_1$ to set
  $d_{p_{\diam + 1}}$ to $\diam$.
\item At the $\diam+1$th round, $p_{\diam}$ executes $HC_2$ to set
  $par_{p_{\diam}}$ to $p_{\diam - 1}$.
\end{itemize}

Hence, we can conclude with the theorem below.

\begin{theorem}
  The worst case stabilization time of $\algoFHC(\diam)$ is $\diam+1$ rounds.
\end{theorem}

\begin{figure}
\begin{centering}
\scalebox{0.7}{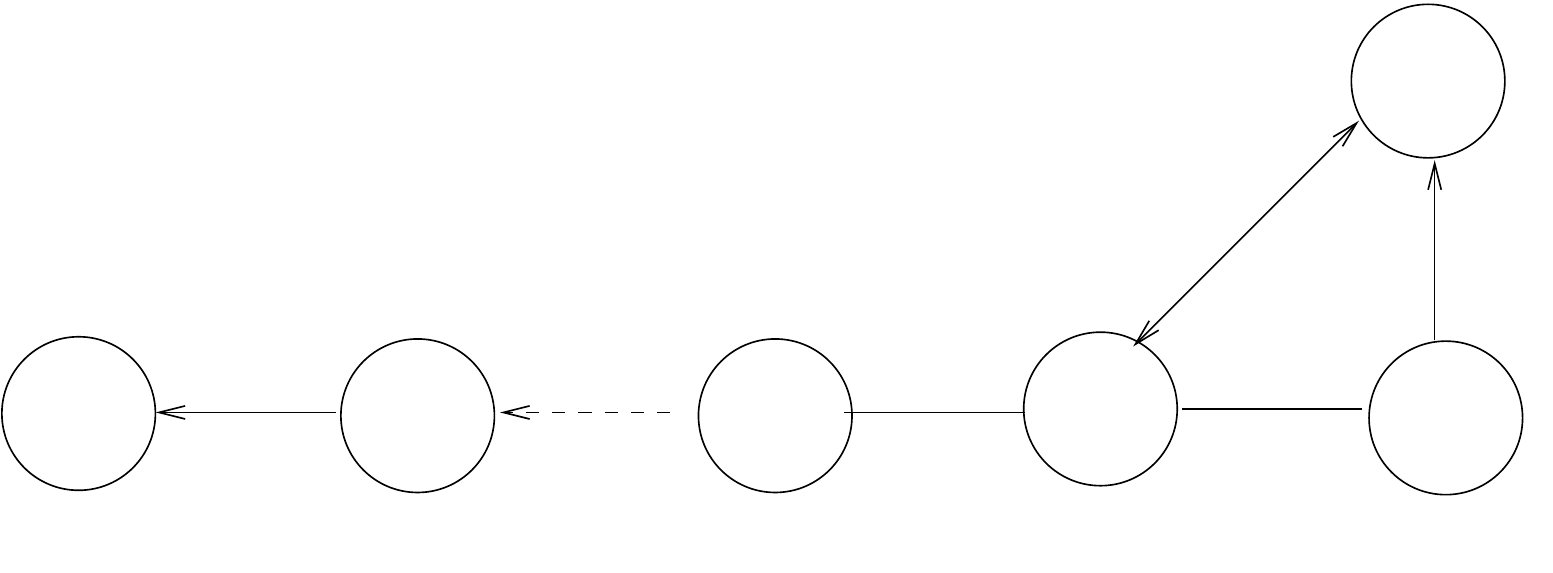}
\caption{Initial configuration of a synchronous execution of
  $\algoFHC(\diam)$ which stabilizes in $\diam+1$ rounds}
\label{fig:execAlgoFHC}
\end{centering}
\end{figure}

\subsection{Algorithms $\algoU$ and $\algoB(D)$}

We now establish that the stabilization time of both $\algoU$ and
$\algoB(D)$ is exactly $\diam$ rounds in the worst case. 

\subsubsection{Upper bound on the Stabilization Time in Rounds
for both $\algoU$ and $\algoB(D)$}
We first establish that the stabilization time of both $\algoU$ and
$\algoB(D)$ is at most $\diam$ rounds in the worst case.
To that goal, we use the predicate $\correcta(p,i)$ defined below:

$$\correcta(p,i) \equiv (\|\Root,p\| > i \Rightarrow d_p > i)$$

$\correcta(p,i)$ means that $d_p$ must be larger than $i$ if the
process $p$ is at a distance larger than $i$ from $\Root$.

We will also use the following sets:
\begin{itemize}
\item Let $\UB(i)$ be the set of  configurations
where every process $p$ satisfies $\correcta(p,i))$. 
\item Let $\A(i) = \LA(i) \cap \UB(i)$.
\end{itemize}
Notice that all configurations belong to  $\UB(0)$.

\begin{lemma}
  If $\A(i)$ is an attractor under $\algoU$ (resp. $\algoB(D)$ where
  $D \geq \diam$) with $0 \leq i < \diam$, then $\LA(i+1)$ is an
  attractor under $\algoU$ (resp. $\algoB(D)$) from $\A(i)$ which is
  reached within at most one round from any configuration of $\A(i)$.
\label{lem:Ai}
\end{lemma}

\begin{proof}
  Let $p_{i+1}$ be a process at distance $i+1$ of $\Root$. By
  definition, $p_{i+1}$ has at least one neighbor $p_i$ at distance
  $i$ of $\Root$.  As $\A(i)$ is an attractor, from any configuration
  of $\A(i)$, the three following conditions hold forever: $(i)$
  $d_{p_i}= i$, $(ii)$ $\mind(p_{i+1}) = i < D$, and $(iii)$ for every
  process $q$, $d_q = i \Rightarrow \| q, \Root \| = i$.

  Let $\gamma \mapsto \gamma'$ be a possible step such that $\gamma$
  is a configuration of $\A(i)$.  We first show that for every process
  $p$, if $\correctb(p_{i+1},i+1)$ holds in $\gamma$, then
  $\correctb(p_{i+1},i+1)$ holds in $\gamma'$.  From $\gamma$,
  $d_{par_{p_{i+1}}} = i$, so $d_{par_{p_{i+1}}}$ is no more modified by
  $(iii)$. In $\gamma$, $d_{p_{i+1}} = i+1$ and $d_{p_{i+1}}$ is no
  more modified by $(ii)$. Hence, $d_{p_{i+1}} = d_{par_{p_{i+1}}}+1$
  and $p$ is disabled forever from $\gamma$. Hence,
  $\correctb(p_{i+1},i+1)$ still holds in $\gamma'$.

  We now show that for every process $p$, $\correctb(p_{i+1},i+1)$
  becomes true within at most one round from any configuration
  $\gamma$ of $\A(i)$. Assume that $d_{p_{i+1}} \neq i+1$ or
  $d_{par_{p_{i+1}}} \neq i$ in $\gamma$. Then, $p_{i+1}$ is enabled
  in $\algoU$ (resp. $\algoB(D)$) until it executes an action, by
  $(i)$ and $(ii)$. Moreover, after $p_{i+1}$ move, we have
  $d_{p_{i+1}} = i +1 = d_{par_{p_{i+1}}} +1$, by $(i)$ and
  $(ii)$. Hence, $\correctb(p_{i+1},i+1)$ becomes true within at most
  one round from $\gamma$.
\end{proof}

\begin{lemma}
  If $\A(i)$ is an attractor under $\algoU$ (resp. $\algoB(D)$ where
  $D \geq \diam$) with $0 \leq i < \diam$, then $\UB(i+1)$ is an
  attractor under $\algoU$ (resp. $\algoB(D)$) from $\A(i)$ which is
  reached within at most one round from any configuration of $\A(i)$.
\label{lem:UBi}
\end{lemma}

\begin{proof}
  Let $p$ be a process such that $\|p,\Root\| > i+1$. In this case, we
  have $i+2 \leq \diam \leq D$.

  In any configuration of $\A(i)$, we have $d_p > i$ and $d_q > i$ for
  any neighbor $q$ of $p$ by definition of $\A(i)$. So, starting from
  any configuration $\gamma$ of $\A(i)$, $\mind(p) > i$ holds forever.

  Hence, if $d_p > i +1$ in $\gamma$, then $d_p > i
  +1$ forever from $\gamma$, which implies that if $\UB(i+1)$ holds in
  $\gamma$, then $\UB(i+1)$ holds forever from $\gamma$.

  Assume now that $d_p = i + 1$ in $\gamma$. Then, $d_p = i + 1 < D$
  and, as $\mind(p) > i$ holds forever from $\gamma$, $p$ is
  continuously enabled from $\gamma$ until it executes either $U_1$ in
  $\algoU$, or $B_j, j \in \{1,2\}$ in $\algoB(D)$. After $p$ move,
  $d_p \geq i +1$. Hence, $\UB(i+1)$ holds within at most one round
  from $\gamma$. 
\end{proof}

\begin{corollary}
  If $\A(i)$ is an attractor under $\algoU$ (resp. of $\algoB(D)$
  where $D \geq \diam$) with $0 \leq i < \diam$, then $\A(i+1)$ is an
  attractor under $\algoU$ (resp. $\algoB(D)$) from $\A(i)$ which is
  reached within at most one round from any configuration of $\A(i)$.
\label{cor:Ai}
\end{corollary}

From the previous corollary and owing the fact that $\A(\diam) =
\LA(\diam)$, we can deduce the following theorem:

\begin{theorem}
The stabilization time of $\algoU$ and $\algoB(D)$ (for every $D \geq \diam$)
 is at most $\diam$ rounds.
\end{theorem}

\subsubsection{Lower Bound on the Stabilization Time in Rounds 
for both $\algoU$ and $\algoB(D)$}
Below, we show that the upper bound is exact for both $\algoU$ and
$\algoB(D)$ when $D \geq \diam$.

Consider first $\algoU$. Let $G = (V,E)$ be any line graph of diameter
$\diam$, {\em i.e.}, $V=\{p_0=\Root, p_1, \ldots, p_\diam\}$ and $E =
\{\{p_i,p_{i+1}\}\ |\ i \in [0..\diam-1]\}$. Consider the initial
configuration where $d_{p_0} = 0$ and $\forall i \in [1..\diam]$,
$d_{p_i} = X$, where $X > \diam$. ($Par$-variables have arbitrary
values). Consider a synchronous execution starting from that initial
configuration. Then, at each round $i$, with $i \in [1..\diam]$, $p_i$
executes $U_1$ to definitely set $d_{p_i}$ to $i$ and $par_{p_i}$ to
$p_{i-1}$. Moreover, $\forall j \in [i+1..\diam]$, $p_j$ increments
$d_{p_j}$ by $U_1$. Hence, after $\diam$ rounds the system is in a terminal
configuration.

%Consider now $\algoB(D)$ with any value $D > \diam$. Consider the same
%graph as for $\algoU$ and the same initial configuration, except that
%$X = D$. The synchronous execution starting from that initial
%configuration is then the following: at each round $i$, with $i \in
%[1..\diam]$, $p_i$ executes $B_1$ to definitely set $d_{p_i}$ to $i$
%and $par_{p_i}$ to $p_{i-1}$.

Consider now $\algoB(D)$ with any value $D \geq \diam$. Consider the
same graph as for $\algoFHC(\diam)$. However, we consider now a
synchronous execution starting from any configuration where:
\begin{itemize}
\item $d_\Root = 0$,
\item $\forall i \in [1..\diam+1], d_{p_i} = D$, 
\item $par_{p_{\diam}} = p_{\diam+1}$, and
\item $par_{p_{\diam+1}} = p_{\diam}$.
\end{itemize}
An example of initial configuration is given in Figure
\ref{fig:execAlgoB}.  The synchronous execution starting from that
configuration then works as follows:
\begin{itemize}
\item In round $i$, with $i \in [1..\diam-1]$, only process $p_i$
  moves. It executes rule $B_1$ to set $d_{p_i}$ to $i$ and
  $par_{p_i}$ to $p_{i-1}$.
\item In round $\diam$, only $p_\diam$ and $p_{\diam+1}$ move. Two
  cases are possibles. Either $D > \diam$ and they both execute $B_1$
  to set $d_{p_\diam}$ (resp. $d_{p_{\diam+1}}$) to $\diam$ and
  $par_{p_\diam}$ (resp. $par_{p_{\diam+1}}$) to $p_{\diam-1}$. Or, $D
  = \diam$ and they both execute $B_2$ to set $par_{p_\diam}$ and
  $par_{p_{\diam+1}}$ to $p_{\diam-1}$.
\end{itemize}

\begin{figure}
\begin{centering}
\scalebox{0.7}{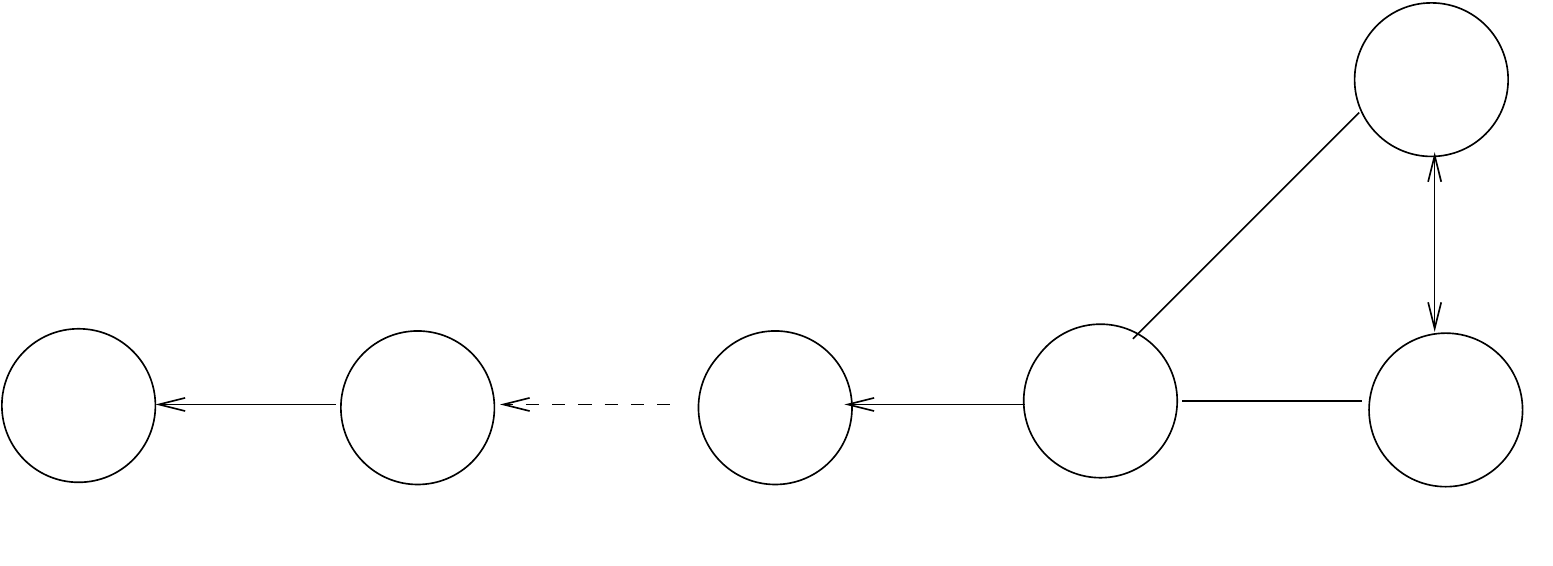}
\caption{Initial configuration of a synchronous execution of
  $\algoB(D)$ which stabilizes in $\diam$ rounds}
\label{fig:execAlgoB}
\end{centering}
\end{figure}

Hence, we can conclude with the theorem below.

\begin{theorem}
  The worst case stabilization time of $\algoU$ and $\algoB(D)$ (with
  $D \geq \diam$) is $\diam$ rounds.
\end{theorem}

\section{Stabilization Time in Steps}\label{sect:step}

In this section, we propose a step complexity analysis of the three
algorithms presented in Section~\ref{sect:algo}.

\subsection{A General Bound}

The theorem below exhibits a trivial upper bound on the stabilization
time in steps of every self-stabilizing algorithm working under an
unfair daemon.

\begin{theorem}
  Let $\mathcal A$ be any self-stabilizing algorithm under an unfair
  daemon,\footnote{The daemon can be central or distributed.} the
  stabilization time of $\mathcal A$ is less than or equal to
  $\prod_{p\in V} |S_p|-2$ steps, where $S_p$ is the set of possible states of
  $p$, for every process $p$.
\end{theorem}
\begin{proof}
  First, the number of possible configurations of $\mathcal A$ is
  $\prod_{p\in V} |S_p|$.  Let $e$ be any execution of $\mathcal A$.
  $\mathcal A$ being self-stabilizing, $e$ contain a maximal prefix of
  finite size $e'= \gamma_i,\gamma_{i+1}\ldots$ where its
  specification is not achieved. Let $e''$ such that $e = e'e''$.

  Assume, by the contradiction, that $\exists k,\ell$ such that $i\leq
  k<\ell$ and $\gamma_k = \gamma_\ell$. Then,
  $(\gamma_{k+1},\ldots,\gamma_\ell)^\infty$ is an infinite execution
  of $\mathcal A$ under the unfair daemon that never stabilized. So,
  $\mathcal A$ is not self-stabilizing under an unfair daemon, a
  contradiction.

  Hence, all configurations of $e'$ are distinct. Moreover, $|e''|
  \geq 1$ and $e'$ and $e''$ have no common configuration. Hence, $e'$
  contains at most $\prod_{p\in V} |S_p|-1$ configurations, and so at
  most $\prod_{p\in V} |S_p|-2$ steps.
\end{proof}

The previous theorem is useless when considering algorithms where at
least one variable as an infinite domain, {\em e.g.}, $\algoU$.  Now,
for $\algoB(D)$ and $\algoHC(D)$, the theorem claims that their
respective stabilization times are less than or equal to
$(n-1)^{\Delta.D}$ steps. This upper bound may appear to be
overestimated at the first glance. However, we will see in the next
subsections that those algorithms are exponential in steps in
the worst case.

\subsection{Algorithm $\algoU$}

Here we consider the unbounded version given in
Subsection~\ref{sub:u}. The following theorem shows that the step
complexity of $\algoU$ cannot be bounded by any function depending on
topological parameters, {\em e.g.}, $n$, $N$, $\diam$, or $D$\ldots

\begin{theorem}
  Let $f$ be any function mapping graphs to integers. There exists a
  graph $G$ and an execution $e$ of $\algoU$ in $G$ such that $e$
  stabilizes in more than $f(G)$ steps.
\end{theorem}
(The following proof is illustrated with Figure~\ref{unbounded}.)\\
\begin{proof}
  Consider a line graph $G$ of 5 nodes, where $R$ is an extremity,
  {\em i.e.}, $G = \{R,p_1,p_2,p_3,p_4\}$ and $E = \{\{R,p_1\}\} \cup
  \{\{p_i,p_{i+1}\}, i \in[1..3]\}$. Let $X \geq f(G)+1$. Assume an
  initial configuration, where $d_{p_1} = d_{p_4} = X$ and $d_{p_2} =
  d_{p_3} = 1$. (The initial values of $par$ variables are arbitrary
  and, by definition, $d_R = 0$.) Initially, all processes, except $R$,
  are enabled. Assume that $p_2$ moves, then in the next
  configuration, $d_{p_2}$ takes value 2 and all processes, except $R$
  and $p_2$, are enabled. Assume that $p_3$ moves, then in the next
  configuration, $d_{p_3}$ takes value 3 and all processes, except $R$
  and $p_3$, are enabled. By alternating activations of $p_2$ and $p_3$
  the system reaches in $X \geq f(G)+1$ steps a configuration where
  $d_R=0$, $d_{p_1} = d_{p_4} = X$, and $d_{p_2} = d_{p_3} = X+1$.
\end{proof}

\begin{figure}
\centering
\includegraphics[scale=0.6]{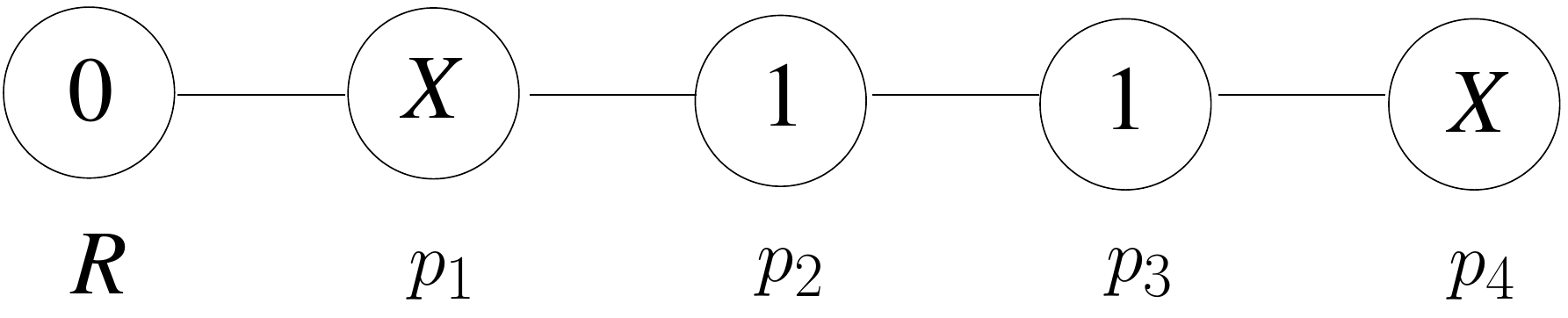}
\caption{Possible initial Configuration of the line of 5 nodes\label{unbounded}}
\end{figure}
 
\subsection{Algorithm $\algoHC(D)$}

In this subsection, we establish that the stabilization times in steps
of both $\algoHC(D)$ and $\algoB(D)$ are exponential in the worst
case.  The lowed bound is based on a family of graphs called
$\mathcal{G}_k$.  For every $k \geq 0$, the graph $\mathcal{G}_k$
contains $4k+3$ processes and has a diameter of $2k+3$.

\begin{definition}[Graph $\mathcal{G}_1$]
\label{def:def1}
Let $\mathcal{G}_1= (V_1, E_1)$ be the undirected graph, where
\begin{itemize}
\item $V_1 = \{f.0, e.1, f.1, h.0, g.1, h.1, \Root\}$ and
\item $E_1 = \{\{\Root, h.0\}, \{h.0, f.0\}, \{f.0, e.1\}, \{e.1,
  f.1\}, \{f.1, h.1\}, \{g.1, e.1\}\}$.
\end{itemize}
\end{definition}

We now consider three classes of configurations for the graph
$\mathcal{G}_1$.  In all consider configurations:
\begin{itemize} 
\item the distance value of 
$g.1$, $h.1$, and $h.0$ are $z-1$,
\item $d_{e.1}=z$ if and only if $par_{e.1} = g.1$, and
\item $d_{f.i}=z$ if and only if $par_{f.i} = h.i$, for $i \in [0,1]$.
\end{itemize}

The three classes of configurations are defined as follows where $x \geq 1$ and $ z >1$:
\begin{itemize}
\item In the configurations of $\Cb_1(x,z)$, the distance value of
  $e.1$ and $f.1$ is $z$, and the distance value of $f.0$ is $x$.
\item In the configurations of $\Cc_1(x,z)$, the distance value of $e.1$  is 
$x$, and the distance value of $f.1$ and $f.0$  is $z$.
\item In the configurations of $\Cd_1(x,z)$, the distance value of
  $e.1$ and $f.0$ is $z$, and the distance value of $f.1$ is $x$.
\end{itemize}
Except otherwise mentioned, all other variables have arbitrary values.
Notice that we have $\Cb(z,z) = \Cc_1(z,z) = \Cd_1(z,z)$.  An
illustrative example of these three types of configurations of
$\mathcal{G}_1$ is given in Figure \ref{fig:fig1}.

\begin{figure}
\begin{centering}
\scalebox{0.5}{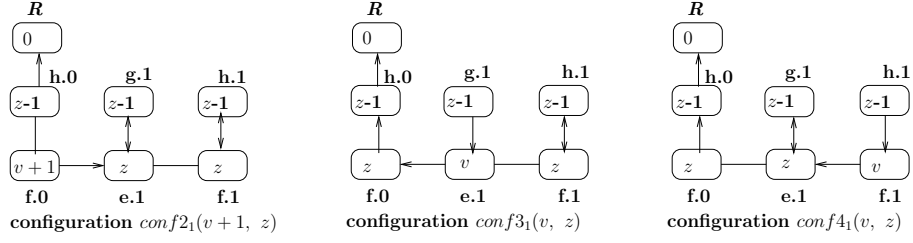}
\caption{Examples of configurations of $\mathcal{G}_1$}
\label{fig:fig1}
\end{centering}
\end{figure}

\begin{observation}
\label{obs:obs1}
Let $v, z, D$ be three integers such that $1 \geq v < z \leq D$.
\begin{itemize}
%\item In any configuration of $\Cc_1(v,z)$, $HC_2$ is enabled at $e.1$
%  and $f.0$.
\item From any configuration of $\Cc_1(v,z)$, 
a configuration of $\Cb_1(v+1,z)$ is reachable  
in a single step of $\algoHC(D)$, where 
 $e.1$ and $f.0$ execute  $HC_2$.
 % \item In a configuration of $\Cb_1(v,z)$, $HC_2$ is enabled at
 %   $e.1$ and $f.0$.
\item From a configuration of $\Cb_1(v,z)$, a configuration  of
$\Cc_1(v+1,z)$ is reached 
in a single step of $\algoHC(D)$, where 
 $e.1$ and $f.0$ execute  $HC_2$. 
\end{itemize}
\end{observation}

\begin{notation}\label{notation1}
  Let $v$ and $z$ be two integers such that $1 \leq v \leq z$ and $z >1$.  Let
  $\const(v,z,1)$ be the maximal number of steps of $\algoHC(D)$
  (with $D \geq z$) to reach a configuration of $\Cc_1(z,z)$ from a
  configuration of $\Cc_1(v,z)$.
\end{notation}

\begin{observation}
  Let $v$ and $z$ be two integers such that $1 \leq v \leq z$ and $z >1$. 
We have
$\const(v+2,z+2,1) = \const(v,z,1)$.
\end{observation}

\begin{lemma}
\label{lem:lem1a}
In $\mathcal{G}_{1}$, for every $1 \leq v \leq z-2$, there is a execution
$e_1(k)$ of $\algoHC(D)$ (with $D \geq z$), starting in a
configuration of $\Cc_{1}(v,z)$ and where only rules $HC_2$ are executed,
which reaches a configuration of $\Cc_{1}(v+2,z)$ in at least $2$
steps.
\end{lemma}
\begin{proof}
Immediate from Observation \ref{obs:obs1}.
\end{proof}

\begin{corollary}\label{cor:co1}~

\begin{itemize}
\item If $1 \leq v \leq z$ and $z >1$ then $\const(v,z+2,1) \geq z-v$.
\item Let $k \geq 1$. In $\mathcal{G}_{1}$, there is a execution of
  $\algoHC(D)$, with $D\geq 2k+3$, which starts in a configuration of
  $\Cc_{1}(1,2k+3)$, contains only executions of rules $HC_2$, and
  reaches a configuration of $\Cc_{1}(2k+3,2k+3)$ in at least $2k+2$
  steps.
\item $\const(1,5,1) = 4$.
\end{itemize}
\end{corollary}

The following definition generalizes Definition~\ref{def:def1}.
 
\begin{definition}[Graph $\mathcal{G}_{i+1}$]
\label{def:defI}
Let $\mathcal{G}_{i+1}= (V_{i+1}, E_{i+1})$ be the undirected graph, where
\begin{itemize}
\item $V_{i+1} = V_i \cup \{e.i+1, f.i+1, g.i+1, h.i+1 \}$ and
\item $E_{i+1} = E_{i} \cup E'_{i+1}$, where $E'_{i+1} = \{\{f.i,
  e.i+1\}, \{g.i+1, e.i+1\}, \{e.i+1, f.i+1\}, \{f.i+1, h.i+1\}\}$.
\end{itemize}
\end{definition}

We mainly consider four classes of configurations for any
graph $\mathcal{G}_{i+1}$. In all consider configurations:
\begin{itemize} 
\item the distance value of $g.i+1$ and $h.i+1$ 
is $z-1$,
\item for every $ j \in [0,i+1]$, $d_{e.j}=z$ if and only if
  $par_{e.j} = g.j$, and
\item for every $ j \in [0,i+1]$, $d_{f.j}=z$ if and only if
  $par_{f.j} = h.j$.
\end{itemize}

The four classes of configurations are then defined as follows where $x \geq 1$ and $ z >1$:
\begin{itemize}
\item In the configurations of $\Ca_{i+1}(x,z)$, the configuration of the subgraph $\mathcal{G}_i$ belongs to
$\Cc_i(x,z)$, the distance value of
$e.i+1$ is $x$ and the distance value of $f.i+1$ is $z$.

\item In the configurations of  $\Cb_{i+1}(x,z)$,  the configuration of the subgraph $\mathcal{G}_i$ belongs to
$\Cd_i(x,z)$, the distance value of
$e.i+1$ and $f.i+1$ is $z$.

\item In the configurations of $\Cc_{i+1}(x,z)$, the configuration of
  the subgraph $\mathcal{G}_i$ belongs to $\Cc_i(z,z)$, the distance
  value of $e.i+1$ is $x$, and the distance value of $f.i+1$ is $z$.

\item In the configurations of $\Cd_{i+1}(x,z)$, the configuration of
  the subgraph $\mathcal{G}_i$ belongs to $\Cc_i(z,z)$, the distance
  value of $e.i+1$ is $z$, and the distance value of $f.i+1$ is $x$.
\end{itemize}
Except otherwise mentioned, all other variables have arbitrary values.
Notice that we have $\Ca_{i+1}(z,z)=\Cb_{i+1}(z,z)$ and
$\Cc_{i+1}(z,z)=\Cd_{i+1}(z,z)$. Some illustrative examples are given in
Figures~\ref{fig:fig2}.

\begin{figure}
\begin{centering}
\scalebox{0.5}{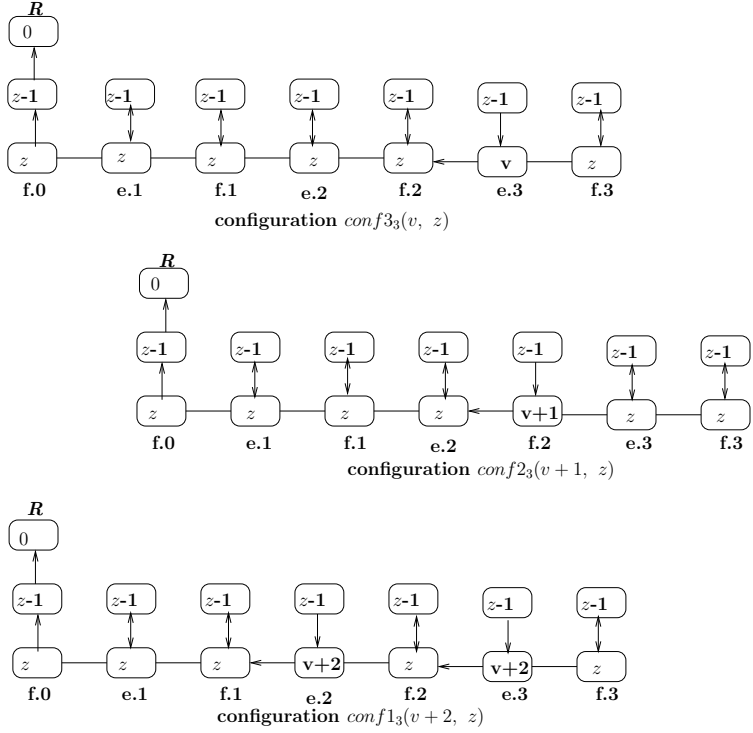}
\caption{Some configurations of $\mathcal{G}_3$}
\label{fig:fig2}
\end{centering}
\end{figure}

\begin{observation}
\label{obs:obsi}
Let $v, z, D$ be three integers such that $1 \leq v < z \leq D$.
\begin{itemize}
%\item In any configuration of $\Cc_{i+1}(v,z)$, $HC_2$ is enabled at $e.i+1$
%  and $f.i$.
\item From any configuration of $\Cc_{i+1}(v,z)$, 
a configuration of $\Cb_{i+1}(v+1,z)$ is reachable  
in a single step of $\algoHC(D)$, where 
 $e.i+1$ and $f.i$ execute  $HC_2$.
 % \item In a configuration of $\Cb_{i+1}(v,z)$, $HC_2$ is enabled at
 %   $e.i+1$, $e.i$, and $f.i$.
\item From a configuration of $\Cb_{i+1}(v,z)$, a configuration  of
$\Ca_{i+1}(v+1,z)$ is reached 
in a single step of $\algoHC(D)$, where 
  $e.i+1$, $e.i$, and $f.i$ execute  $HC_2$. 
\end{itemize}
\end{observation}

The notation below generalizes Notation~\ref{notation1}.

\begin{notation}
  Let $v$ and $z$ be two integers such that $v \leq z$. Let $i\geq 1$. Let
  $\const(v,z,i)$ be the maximal number of steps of $\algoHC(D)$
  (with $D \geq z$) to reach a configuration of $\Cc_i(z,z)$ from a
  configuration of $\Cc_i(v,z)$ in the graph $\mathcal{G}_j$ with $j \geq i$.
\end{notation}

\begin{observation}\label{obs:i+1}
  Let $v$ and $z$ be two integers such that $1 \leq v \leq z$ and $1 < z$. 
We have
  $\const(v+2,z+2,i) = \const(v,z,i)$ and $\const(v, z,i+1) \geq
  \const(v,z,i)$.
\end{observation}

\begin{lemma}
\label{lem:lemia}
Let $z$ and $D$ be two integers such that $z \leq D$. Let $i \geq 1$.
In the graph $\mathcal{G}_{j+1}$ with $j \geq i$, for every $1 \leq v
\leq z-2$, there is an execution of $\algoHC(D)$ starting from a
configuration of $\Cc_{i+1}(v,z)$, where only rules $HC_2$ are
executed, which reaches a configuration of $\Cc_{i+1}(v+2,z)$ in at
least $\const(v+2, z,i)+2$ steps.
\end{lemma}
\begin{proof}
  From a configuration of $\Cc_{i+1}(v,z)$, a configuration of
  $\Ca_{i+1}(v+2,z)$ is reached in two steps of
  $\algoHC(D)$ where only rules $HC_2$ are executed, by Observation
  \ref{obs:obsi}.  From a configuration of $\Ca_{i+1}(v+2,z)$, a
  configuration of $\Cc_{i+1}(v+2,z)$ is reached in at least
  $\const(v+2,z,i)$ steps of $\algoHC(D)$ where processes of the
  subgraph $\mathcal{G}_{i}$ only execute rules $HC_2$ (according to
  the definition of $\const(v,z,i)$).
\end{proof}

From Observation~\ref{obs:i+1} and Lemma~\ref{lem:lemia}, we can
deduce the following corollary.

\begin{corollary}
\label{cor:coric} ~

\begin{itemize}
\item Let $1 \leq v \leq z$ and $1 < z$.  
$\const(v,z+2,i+1) \geq 2+ \const(v,z,i)+
  \const(v,z,i+1)$.
\item Let $j \geq k$.  In the graph $\mathcal{G}_{j}$, there is a
  execution $e_{k}(j)$ of $\algoHC(D)$, with $D \geq 2j+3$, which
  starts in a configuration of $\Cc_{k}(1,2j+3)$, contains only
  executions of rules $HC_2$, and reaches a configuration of
  $\Cc_{k}(2j+3,2j+3)$.
\end{itemize}
\end{corollary}

\begin{notation}
$\nbTotal(2k+3) = \sum _{\ell=1}^{k}\const(1,2k+3,\ell)$. 
\end{notation}

\begin{definition}
  let $k \geq 1$.  Let $e^k$ be the execution of $\algoHC(D)$, with $D
  \geq 2k+3$, in the graph $\mathcal{G}_{k}$ defined as follows: $e^k$
  is the concatenation of $e_1(k) \ldots e_{k}(k)$.
 \end{definition}

 By definition, $e^k$ contains at least $\nbTotal(2k+3)$ steps,
 moreover those steps are only made of rules $HC_2$'s executions.

\begin{theorem}
\label{theo:nbTotal}
For all $k >1$, 
$\nbTotal(2k+3) = 2.\nbTotal(2k+1)+ 2k+ \const(3,2k+3,k)$.
\end{theorem}
\begin{proof}
$\nbTotal(2k+3) = \sum _{\ell=1}^{k}\const(1,2k+3,\ell)$.\\
We have: 
\begin{itemize}
\item $\const(1,2k+3,1) = 2+ \const(1,2k+1,1)$, by Corollary \ref{cor:co1}. 
\item $\const(1,2k+3,k)
= 2 + \const(1,2k+1,k-1) +\const(3,2k+3,k) $, by Corollary \ref{cor:coric} and Observation~\ref{obs:i+1}. 
\item $\const(1,2k+3,\ell) = 2 + \const(1,2k+1,\ell-1) + \const(1,2k+1,\ell)$ 
for $\ell \in [2,k-1]$, by Corollary \ref{cor:coric}.
\end{itemize}
\noindent So, we can conclude that $\nbTotal(2k+3) = \sum
_{\ell=1}^{k-1}( 2.\const(1,2k+1,\ell))+ 2k+ \const(3,2k+3,k)$.
\end{proof}

The following corollary establishes a lower bound on the number of
steps of $e^k$ which is exponential on the graph diameter:
$2^{\frac{\diam -1}{2}}$.

\begin{corollary}
For all $k \geq 1$, 
$\nbTotal(2k+3) \geq 2^{k-1}.\const(1,5,1) \geq 2^{k+1}$.
\end{corollary}

We now propose a tighter bound on $\nbTotal(2k+3)$.

\begin{lemma}
\label{lem:lemVi2}
For all $k \geq 1$, 
$\const(3,2k+3,k) \geq 2.(2^{k}-1)$.
\end{lemma}
\begin{proof}
By induction. 
\begin{description}
\item[Base Case:] Let $k=1$. We have $\const(3,5,1) \geq 2 = 2.(2^{1}-1)$, by Corollary \ref{cor:co1}. 
\item[Induction Hypothesis:]
Assume that $\const(3,2k+3,k) \geq 2.(2^{k}-1)$. 
\item[Induction Step:] By Corollary \ref{cor:coric}, we have:
$$\const(3,2k+5,k+1)\geq 2 + \const(3,2k+3,k)+ \const(3,2k+3,k+1)$$
By Observation~\ref{obs:i+1}, we have:
$$\const(3,2k+5,k+1)\geq 2 + 2.\const(3,2k+3,k)\geq2+ 4.(2^{k}-1)
= 2.(2^{k+1}-1)$$
\end{description}
\end{proof}

\begin{lemma}\label{lem:bound}
For all $k \geq 1$, 
$\nbTotal(2k+3) \geq (2k+2)(2^{k}-1)$.
\end{lemma}
\begin{proof}
By induction.
\begin{description}
\item[Base Case:] Let $k=1$. 
$\nbTotal(5) = \const(1,5,1) = 4$, by Corollary \ref{cor:co1}.
\item[Induction Hypothesis:] Assume that
$\nbTotal(2k+3) \geq (2k+2)(2^{k}-1)$ for $k  \geq 1$.
\item[Induction Step:] ~

\begin{itemize}
\item $\nbTotal(2k+5) =2.\nbTotal(2k+3)+ 2k+2+ \const(3,2k+5,k)$, by Theorem \ref{theo:nbTotal}.
\item $\nbTotal(2k+5) \geq 2.(2k+2)(2^{k}-1)+ 2k+2+ \const(3,2k+5,k)$, by induction hypothesis.
\item $\nbTotal(2k+5) \geq 2.(2k+2)(2^{k}-1)+ 2k+2+ 2.(2^{k}-1)$, by Lemma \ref{lem:lemVi2}.
\item $\nbTotal(2k+5) \geq (2k+2)(2^{k+1}-2) + 2.2^{k+1} +2k$.
\item $\nbTotal(2k+5) \geq (2k+4)2^{k+1}-2.(2k+2) +2k$.
\end{itemize}
So, we conclude that
$\nbTotal(2k+5) \geq (2k+4)2^{k+1}-(2k+4)$.
\end{description}
\end{proof}

\begin{theorem}\label{theo:bound}
  Let $n \geq 7$. Let $k$ the maximum integer such that $n = 4k+3+y$
  with $y \geq 0$. For every $D \geq 2k+3$, there is an execution of
  $\algoHC(D)$ which stabilizes in at least $(2k+2)(2^{k}-1)$ steps
  containing only executions of rules $HC_2$ in an $n$-node graph of
  diameter at most $2k+4$.
\end{theorem}
\begin{proof}
  Let $\mathcal{G}_{k} = (V_k, E_k)$.  Let $\mathcal{G}_{k}^\prime =
  (V_k \cup {v_1,\ldots,v_y}, E_k \cup \{\{v_i,\Root\}, i \in
  [1..y]\})$.  Since $\mathcal{G}_{k}$ has diameter $2k+3$,
  $\mathcal{G}_{k}^\prime$ has at most diameter $2k+4$. Since
  $\mathcal{G}_{k}$ contains $4k+3$ nodes, $\mathcal{G}_{k}^\prime$
  contains $n$ nodes. Finally, nodes $v_1, \ldots, v_y$ are only
  neighbors of $\Root$ whose state is constant. So, $v_1, \ldots, v_y$
  have no impact on the behavior of nodes of $\mathcal{G}_{k}$. Hence,
  we can apply Lemma~\ref{lem:bound} and we are done.
\end{proof}

\begin{corollary}
  For every $n$-node graph $G$, $\algoHC(n)$, {\em i.e.}, the algorithm
  proposed in~\cite{HC92}, stabilizes in $\Omega(2^\diam)$ steps,
  where $\diam$ is the diameter of $G$.
\end{corollary}

\subsection{Algorithm $\algoB(D)$}

Theorem~\ref{theo:bound} exhibits an execution exponential in steps
where only rules $HC_2$ are executed. So, this execution is also an
execution of $\algoFHC(D)$ ({\em i.e.}, the fast implementation of
$\algoHC(D)$).  Moreover, this is also a definition of $\algoB(D)$, by
Lemma \ref{lem:fast-impl1}). Hence, we can conclude with the following
theorem:

\begin{theorem}
  Let $n \geq 7$. Let $k$ the maximum integer such that $n = 4k+3+y$
  with $y \geq 0$. For every $D \geq 2k+3$, there is an execution of
  $\algoB(D)$ (resp.  $\algoFHC(D)$) which stabilizes in at least
  $(2k+2)(2^{k}-1)$ steps in an $n$-node graph of diameter at most
  $2k+4$.
\end{theorem}

\section{Conclusion and Perspective}\label{sect:ccl}

In this paper, we revisited two fundamental results of the
self-stabilizing literature~\cite{HC92,DIM93}. More precisely, we
proposed three silent self-stabilizing BFS spanning tree algorithms
working in the composite atomicity model inspired from the solutions
proposed in~\cite{HC92,DIM93}: Algorithms $\algoU$, $\algoB(D)$, and
$\algoHC(D)$. We then presented a deep study of these algorithms.  Our
results are related to both correctness and complexity.

Concerning the correctness part, we proposed in particular a new,
simple, and general proof scheme to show the convergence of silent
algorithms under the distributed unfair daemon. We believe that our
approach, based on process partitioning, is versatile enough to be
applied in the convergence proof of many other silent algorithms.

Concerning the complexity part, our analysis notably shows that the
Huang and Chen's algorithm~\cite{HC92} stabilizes in $\Omega(n)$
rounds (where $n$ is the size of the network), while it confirms that
the stabilization time in rounds of the Dolev {\em et al}'s
algorithm~\cite{DIM93} is optimal (exactly $\diam$ rounds in the worst
case). Finally, our analysis reveals that the stabilization time in
steps of $\algoU$ cannot be bounded, while the stabilization time of
both $\algoB(D)$ and $\algoHC(D)$ can be exponential in $\diam$, the
diameter of the network. Our results must be put in perspective with
the complexities of the silent BFS construction proposed
in~\cite{CRV11}, which stabilizes in $O(\diam^2)$ rounds and $O(n^6)$
steps, respectively. This suggests the existence of a trade-off
between the complexity in rounds and steps for the silent construction
of a BFS tree. This conjecture would have to be investigated in future
works.

\bibliographystyle{unsrt}
\bibliography{biblio}

\end{document}